\documentclass[letterpaper]{article} 
\usepackage{aaai24}  
\usepackage{times}  
\usepackage{helvet}  
\usepackage{courier}  
\usepackage[hyphens]{url}  
\usepackage{graphicx} 
\usepackage{natbib}  
\usepackage{caption} 
\usepackage{booktabs}       
\usepackage{amsfonts}       
\usepackage{nicefrac}       
\usepackage{microtype}      
\usepackage{xcolor}         

\usepackage[english]{babel}
\usepackage[utf8]{inputenc}
\usepackage{acro}
\usepackage{adjustbox}
\usepackage[linesnumbered,ruled,lined,noend]{algorithm2e}
\makeatletter
    \patchcmd\algocf@Vline{\vrule}{\vrule \kern-0.4pt}{}{}
    \patchcmd\algocf@Vsline{\vrule}{\vrule \kern-0.4pt}{}{}
\makeatother
\SetKwComment{Hline}{}{\vspace{-3mm}\textcolor{gray}{\hrule}\vspace{1mm}}
\definecolor{darkgrey}{gray}{0.3}
\definecolor{commentcolor}{gray}{0.3}
\SetKwComment{Comment}{\color{commentcolor}$\triangleright$\ }{}
\SetCommentSty{}
\SetNlSty{}{\color{darkgrey}}{}
\SetKwProg{Fn}{function}{}{}
\SetKwProg{Subr}{subroutine}{}{}

\makeatletter
\let\cref@old@stepcounter\stepcounter
\def\stepcounter#1{%
  \cref@old@stepcounter{#1}%
  \cref@constructprefix{#1}{\cref@result}%
  \@ifundefined{cref@#1@alias}%
    {\def\@tempa{#1}}%
    {\def\@tempa{\csname cref@#1@alias\endcsname}}%
  \protected@edef\cref@currentlabel{%
    [\@tempa][\arabic{#1}][\cref@result]%
    \csname p@#1\endcsname\csname the#1\endcsname}}
\makeatother

\usepackage{amsfonts}
\usepackage{amsmath}
\usepackage{amsthm}
\usepackage{blindtext}
\usepackage{bm}
\usepackage{centernot}
\usepackage{changepage}
\usepackage{csquotes}
\usepackage{float}
\usepackage{lipsum}
\usepackage{mathtools}
\usepackage{makecell}
\usepackage{multicol}
\usepackage{ragged2e}
\usepackage{rotating}
\usepackage{scalerel}
\usepackage{stackengine}
\usepackage{subcaption}
\usepackage{tablefootnote}
\usepackage{tabularx}
\usepackage{thm-restate}
\usepackage{thmtools}
\usepackage{wasysym}
\usepackage{xargs}
\usepackage{xspace}
\usepackage{xparse}
\usepackage{xpatch}
\usepackage{url}            
\usepackage{nicefrac}       
\usepackage{todonotes}

\graphicspath{
  {./},
  {./images/},
  {../images/},
  {../../images},
}

\newtheorem{theorem}{\protect\theoremname}

\providecommand{\corollaryname}{Corollary}
\providecommand{\claimname}{Claim}
\providecommand{\definitionname}{Definition}
\providecommand{\lemmaname}{Lemma}
\providecommand{\notationname}{Notation}
\providecommand{\remarkname}{Remark}
\providecommand{\problemname}{Problem}
\providecommand{\propositionname}{Proposition}
\providecommand{\examplename}{Example}
\providecommand{\theoremname}{Theorem}
\providecommand{\conjecturename}{Conjecture}
\providecommand{\experimentname}{Experiment}

\DeclareMathAlphabet{\mathpzc}{OT1}{pzc}{m}{it}
\DeclareMathSymbol{\shortminus}{\mathbin}{AMSa}{"39}

\newcommand{\R}{\mathbb{R}}
\newcommand{\mc}{\mathcal}

\newcommand{\vv}[1]{%
 \bm{#1}
}

\newcommand{\norm}[1]{\left\lVert#1\right\rVert_2}  
\newcommand{\inftynorm}[1]{\left\lVert#1\right\rVert_\infty}  

\newcommand{\inner}[2]{\langle #1, #2 \rangle}    

\newcommand{\expl}{\mathpzc{expl}\!}              

\newcommand{\defword}[1]{\textbf{\boldmath{#1}}}

\newcommand{\bandit}{\ensuremath{m}}
\newcommand{\task}{\ensuremath{g}}
\newcommand{\tasks}{\ensuremath{G}}
\newcommand{\predictor}{\ensuremath{\pi}}
\newcommand{\prediction}{\ensuremath{\vv{p}}}

\newcommand{\reward}{\ensuremath{\vv{x}}}
\newcommand{\regret}{\ensuremath{\vv{r}}}
\newcommand{\cumregret}{\ensuremath{\vv{R}}}
\newcommand{\predregret}{\ensuremath{\vv{\xi}}}
\newcommand{\maxutildiff}{\ensuremath{\Delta_{\textnormal{max}}}}

\newcommand*\cleartoleftpage{%
  \clearpage
  \ifodd\value{page}\hbox{}\newpage\fi
}

\definecolor{insignwin}{rgb}{0.5,0.5,0.5}
\definecolor{insignlose}{rgb}{0.5,0.5,0.5}
\definecolor{win}{rgb}{0,0,0}
\definecolor{lose}{rgb}{0,0,0}

\newcommand{\strategy}{\vv{\sigma}}

\usepackage{tikz} 
\usetikzlibrary{shapes} 
\usetikzlibrary{calc} 

\newlength{\nodesize}
\colorlet{chance_color}{black}
\colorlet{pl0_color}{chance_color}
\colorlet{chance_text}{white}
\colorlet{pl1_color}{magenta!50}
\colorlet{pl2_color}{green!50!lime!60}
\tikzset{
basenode/.style = {draw,
inner sep = 0.1em,
minimum size = \nodesize
},
playernode/.style={basenode,
shape = regular polygon,
regular polygon sides = 3
},
pl1/.style={playernode, fill=pl1_color},
pl2/.style={playernode, fill=pl2_color, shape border rotate=180},
chance/.style = {basenode,
fill=pl0_color, text=chance_text,
circle,
minimum size=0.7*\nodesize,
},
terminal/.style = {basenode,
draw=none,
outer sep=0,
minimum size = 0.6\nodesize
}
}

\title{Learning not to Regret}

\author{
    David Sychrovsk\'{y}$^{1,2}$, 
    Michal \v{S}ustr$^{2,5}$,  
    Elnaz Davoodi$^3$,\\
    Michael Bowling$^4$, 
    Marc Lanctot$^3$, 
    Martin Schmid$^{1,5}$
}
\affiliations{
    $^1$Department of Applied Mathematics, Charles University\\
    $^2$Artificial Intelligence Center, Czech Technical University\\
    $^3$Google DeepMind\\
    $^4$Department of Computing Science, University of Alberta\\
    $^5$EquiLibre Technologies
    
    {\tt sychrovsky@kam.mff.cuni.cz,
    michal.sustr@aic.fel.cvut.cz,
    schmid@equilibretechnologies.com}
}

\begin{document}

\maketitle

\begin{abstract}



The literature on game-theoretic equilibrium finding predominantly focuses on single games or their repeated play.
Nevertheless, numerous real-world scenarios feature playing a game sampled from a distribution of similar, but not identical games, such as playing poker with different public cards or trading correlated assets on the stock market. 
As these similar games feature similar equilibra, we investigate a way to accelerate equilibrium finding on such a distribution.
We present a novel ``learning not to regret'' framework, enabling us to meta-learn a regret minimizer tailored to a specific distribution. 
Our key contribution, Neural Predictive Regret Matching, is uniquely meta-learned to converge rapidly for the chosen distribution of games, while having regret minimization guarantees on any game.
We validated our algorithms' faster convergence on a distribution of river poker games. 
Our experiments show that the meta-learned algorithms outpace their non-meta-learned counterparts, achieving more than tenfold improvements.

\end{abstract}

\section{Introduction}
\label{sec: intro}

Regret minimization, a fundamental concept in online convex optimization and game theory, plays an important role in decision-making algorithms~\citep{nisan2007algorithmic}.
In games, a common regret minimization framework is to cast each player as an independent online learner. This learner interacts repeatedly with the game, which is represented by a black-box environment and encompasses the strategies of all other players or the game’s inherent randomness. When all the learners employ a regret minimizer, their average strategy converges to a coarse correlated equilibrium~\citep{hannan1957approximation, hart2000simple}.
Furthermore, in two-player zero-sum games, the average strategy converges to a Nash equilibrium~\citep{nisan2007algorithmic}. 
Regret minimization has become the key building block of many algorithms for finding Nash equilibria in imperfect-information games~\citep{bowling2015heads,DeepStack,brown2018superhuman,brown2020combining,Pluribus,schmid2021player}.

While these algorithms made progress in single game playing, in many real-world scenarios, players engage in more than just one isolated game. For instance, they might play poker with various public cards, solve dynamical routing problems, or trade correlated assets on the stock market. 
These games, while similar, are not identical and can be thought of as being drawn from a distribution. Despite its relevance, this setting has been largely unexplored, with a few recent exceptions such as \cite{harris2022meta,zhang2022no}.

In this work, we shift focus to this distributional setting. The black-box environment which the learners interact with corresponds to a game \emph{sampled from a distribution}. This perspective aligns with the traditional regret minimization framework, but with an added twist: the game itself is sampled. Our goal is to reduce the \emph{expected} number of interactions needed to closely approximate an equilibrium of the sampled game. This is crucial both for online gameplay and offline equilibrium learning, as fewer steps directly translate to a faster algorithm.

In either the single-game or distributional settings, the worst-case convergence 
of regret minimizers against a strict adversary cannot occur at a rate faster than $O(T^{-1/2})$~\citep{nisan2007algorithmic}.
However in practice, algorithms often converge much faster than the worst-case bound suggests. 
Consider CFR$^+$ \citep{tammelin2014solving}, which empirically converges at the rate of $O(T^{-1})$ in poker games\footnote{
The strong empirical performance of the algorithm was one of the key reasons behind essentially solving Limit Texas Holdem poker, one of the largest imperfect information games to be solved to this day~\citep{bowling2015heads}.
CFR$^+$ required only $1,579$ iterations to produce the final strategy, far less than what the worst-case bound suggests.} 
despite having the same $O(T^{-1/2})$ worst case guarantees~\citep{burch2018time}.
Another example of variations in practical performance is discounted CFR with three parameters ($DCFR_{\alpha, \beta, \gamma}$), where the authors reported that they ``found the optimal choice of $\alpha, \beta$ and $\gamma$ varied depending on the specific game'' \citep{brown2019solving}.

These empirical observations are in line with no-free lunch theorems for optimization, which state that no learning algorithm can dominate across all domains \citep{wolpert1997no}.
Thus to improve performance on a domain, it is necessary to use a specialized algorithm, at the expense of deteriorating the performance outside of this domain.

A popular approach to find such algorithms is the meta-learning paradigm, namely a variant of ``learning to learn''~\citep{andrychowicz2016learning}. 
In the meta-learning framework, one learns the optimization algorithm itself.
The simplest approach is to directly parametrize the algorithm with a neural network, and train it to minimize regret on the distribution of interest.
While the meta-learned network can quickly converge in the domain it has been trained on (e.g. poker games), it can be at the cost of performance (or even lack of convergence)  out-of-distribution.
This is because the neural network is not necessarily a regret minimizer.

To provide the convergence guarantees, we introduce meta-learning within the predictive regret framework~\citep{farina2021faster}.
Predictive regret minimization has convergence guarantees regardless of the prediction, while a better prediction guarantees lower regret, and a perfect prediction results in zero regret~\citep{farina2021faster}.
This results in an algorithm that combines the best of both worlds -- fast convergence in the domain in question while providing general convergence guarantees.

A particularly interesting application of our approach is when the resulting regret minimizer is used in an online search algorithm~\citep{DeepStack, brown2018superhuman, schmid2021player}.
When the agent is deployed to face an opponent in chess, poker or other games, it has a limited time to make a decision. 
The agent needs to minimize regret within its search tree as quickly as possible --- that is, with as few iterations as possible.
This is because a single iteration evaluates the leaf nodes of a search tree using a value function, which is typically represented by a slow-to-compute neural network.
In this context, the critical measure is the speed during the actual deployment time and online search, that is, when facing the opponent.
The offline computation is typically used to learn high quality value functions to be used within search and can take even long time.
With our method, one can now also use the offline computation to meta-learn the regret minimizer itself, resulting in substantially faster convergence during the play time.

In experiments, we first evaluate our algorithms on a distribution of matrix games to understand what the algorithms learn.
Next, we turn our attention to search with value functions in a sequential decision setting.
We show that for a distribution over river poker games, our meta-learned algorithms outpace their non-meta-learned counterparts, achieving more than tenfold improvements.

\section{Prior Work}

Regret minimization is a powerful framework for online convex optimization \citep{zinkevich2003online}, with regret matching as one of the most popular algorithms in game applications \citep{hart2000simple}.
Counterfactual regret minimization allows to use that framework in sequential decision making, by decomposing the full regret to individual states~\citep{zinkevich2008regret}.
A recently introduced extension of regret matching, the predictive regret matching~\citep{farina2021faster} was shown to significantly outperform prior regret minimization algorithms in self-play across a large selection of games.
The authors also provided a close connection between the prediction and the regret, which offers additional insight into the algorithm and is a clear inspiration for our work. 

Meta-learning has a long history when used for optimization~\citep{schmidhuber1992learning, schmidhuber1993neural, thrun1996explanation, andrychowicz2016learning}.
This work rather considers meta-learning in the context of regret minimization.
Many prior works explored modifications of regret matching to speed-up its empirical performance in games, such as
CFR+ \cite{tammelin2014solving},
DCFR \cite{brown2019solving},
Lazy-CFR \cite{zhou2018lazy},
ECFR \cite{li2020solving} or
Linear CFR \cite{brown2019deep}. 
However, as the no-free lunch theorems for optimization state, no (learning) algorithm can dominate across all domains~\cite{wolpert1997no}.
Therefore, to improve performance on a specific domain, it is necessary to use a specialized algorithm, at the expense of deteriorating the performance outside of this
domain.

We thus turn to meta-learning the regret minimizers. It was shown that similar games have similar equilibria, justifying the use of meta-learning in games to accelerate equilibrium finding~\citep{harris2022meta}.
A key difference between our and prior works is that they primarily consider settings where the game utilities come from a distribution, rather than sampling the games themselves.
Thus, one of their requirements is that the strategy space itself must be the same.
In~\citep{azizi2022meta}, they consider bandits in Bayesian settings.
In~\citep{harris2022meta}, the authors ``warm start'' the initial strategies from a previous game, making the convergence provably faster. This approach is ``path-dependant'', in that it depends on which games were sampled in the past.
Both works are fundamentally different from ours, as they use meta-learning online, while we are making meta-learning preparations offline.

To our best knowledge, the most similar to our offline meta-learning setting is AutoCFR~\citep{xu2022autocfr}. They are not restricted to the same strategy spaces in games like previous works, as they use evolutionary search for an algorithm that is local to each decision state.
They search over a combinatorial space, defined by an algebra which generalizes CFR family of algorithms, to find an algorithm that performs well across many games.
Our approach rather learns a neural network via gradient descent to perform the regret minimization, allowing us to learn any function representable by the network architecture. 
Furthermore, unlike AutoCFR, we provide strong regret minimization guarantees.

We also give a quick overview of the recent work on search with value functions, as we use regret minimization in this context in our experiments.
The combination of decision-time search and value functions has been used in the remarkable milestones where computers bested their human counterparts in challenging games --- DeepBlue for Chess~\citep{campbell2002deep} and AlphaGo for Go~\citep{silver2016mastering}.
This powerful framework of search with (learned) value functions has been extended to imperfect information games~\citep{schmid2021search}, where regret minimization is used within the search tree.
Regret minimization has quickly become the underlying equilibrium approximation method for search~\citep{DeepStack, brown2018superhuman, zarick2020unlocking, serrino2019finding, brown2020combining, schmid2021player}.

\section{Background}
We begin by describing the regret minimisation framework~\cite{nisan2007algorithmic}. An \defword{online algorithm} $\bandit$ for the regret minimization task repeatedly interacts with an unknown \defword{environment} $\task$ through available actions $A$, receiving a vector of per-action rewards $\reward$.
The goal of regret minimization algorithm is then to maximize its hindsight performance (i.e. to minimize regret).


Formally, at each step $t\le T$, the algorithm submits a \defword{strategy} $\strategy^t$ from a probability simplex $\Delta^{|A|}$ and observes the subsequent \defword{reward} $\reward^t \in \R^{|A|}$ returned from the environment $\task$.
The rewards are computed with an unknown function concave in $\strategy$ and are bounded. We denote by $\maxutildiff$ the difference between the highest and lowest reward the environment can produce.
The difference in reward obtained under $\strategy^t$ and any fixed action strategy is measured by the instantaneous \defword{regret} $\regret(\strategy^t, \reward^t) = \reward^t - \inner{\strategy^t}{\reward^t} \vv{1}$.
 A sequence of strategies and rewards, submitted by algorithm $\bandit$ and returned by environment $\task$, up to a horizon $T$, is
\begin{equation}\label{eq:strategy-reward-sequence}
    \reward^0 \rightarrow \strategy^1 \rightarrow
    \reward^1 \rightarrow \strategy^2 \rightarrow
    \dots \rightarrow
    \reward^{T-1}\rightarrow \strategy^{T} \rightarrow
    \reward^{T},
\end{equation}
where we set $\reward^0 = \mathbf{0}$ for notational convenience (see also Figure~\ref{fig: strategy reward sequence}).
The \defword{cumulative regret} over the entire sequence is
\begin{equation*}
    \cumregret^T = 
    \sum_{t=1}^T \regret(\strategy^{t}, \reward^t).
\end{equation*}

The algorithm $\bandit$ is a regret minimizer, if the \defword{external regret}
$R^{\text{ext}, T} = \inftynorm{\vv{R}^T}$
grows sublinearly in $T$ for an arbitrary sequence of rewards $\{\reward^t\}_{t=1}^T$. Then the average strategy $\overline{\strategy}^t = \frac{1}{t}\sum_{\tau=1}^t \strategy^\tau$ converges to a coarse correlated equilibrium~\citep{nisan2007algorithmic}. 

\medskip

Finally, we define \defword{exploitability} of a strategy $\strategy$ (i.e. the gap from a Nash equilibrium) as
\begin{equation*}
    \expl(\strategy) = 
    \max_{\strategy^*} \min_{\reward} \inner{\strategy^*}{\reward(\strategy^*)}
    -\min_{\reward} \inner{\strategy}{\reward(\strategy)},
\end{equation*}
where $\reward(\strategy)$ is the reward vector admissible by the environment as a response to strategy $\strategy$.
Note that this exactly corresponds to the standard definition of exploitability of a player's strategy in a two-player zero-sum game when playing with an environment controlled by an adversary.

\section{Learning not to Regret}

We first describe the meta-learning framework for regret minimization.
Then we introduce two variants of meta-learned algorithms, with and without regret minimization guarantees.

\subsection{Meta-Learning Framework}
On a distribution of regret minimization tasks $\tasks$, we aim to find an online algorithm $\bandit_\theta$ with some parameterization $\theta$ that efficiently minimizes the expected external regret after $T$ steps. The expected external regret of $\bandit_\theta$ is
\begin{align}
    \label{eq: loss}
    \mc{L}(\theta) &=
    \mathop{\mathbb{E}}_{\task \sim \tasks}
    \left[
    R^{\text{ext},T}\right]
    =
    \mathop{\mathbb{E}}_{\task \sim \tasks}
    \left[
    \max_{a\in A}
    \sum_{t=1}^T \regret_a\left(\strategy^{t}_\theta, \reward^t \right)\right],
\end{align}
where $\strategy^{t}_\theta$ is the strategy selected at step $t$ by the online algorithm $\bandit_\theta$.
We train a recurrent neural network parameterized by $\theta$ to minimize \eqref{eq: loss}. 
By utilizing a recurrent architecture we can also represent algorithms that are history and/or time dependent. This dependence is captured by a hidden state $\vv{h}$ of the recurrent network. See Section~\ref{sec:experiments} for details. 

The choice to minimize external regret in particular is arbitrary. This is because the rewards $\{\reward^\tau\}_{\tau=1}^{T}$ that come from the environment are constant\footnote{This is because the environment is a black-box, i.e. how the reward depends on the chosen strategy is unknown to $\bandit_\theta$.\label{fn: detach reward}} w.r.t. $\theta$ and the derivative of any element of the cumulative regret vector $\vv{R}^T$ is thus the same, meaning\footnote{When recurrent architecture is used, the strategy also depend on strategies used in previous steps, intruding extra terms.}
\begin{equation*}
    \frac{\partial\mc{L}}{\partial \strategy_\theta^t} = 
    \frac{\partial}{\partial \strategy_\theta^t}
    \mathop{\mathbb{E}}_{\task \sim \tasks}
    \left[
    \sum_{\tau=1}^T -\inner{\strategy^{\tau}_\theta}{\reward^\tau}\right] = 
    -\mathop{\mathbb{E}}_{\task \sim \tasks}
    \left[\reward^t\right].
\end{equation*}
Consequently, if the objective~\eqref{eq: loss} is reformulated using other kinds of regrets, it would result in the same meta-learning algorithm. 
This is because regrets measure the difference between reward accumulated by some fixed strategy, and by the algorithm $\bandit_\theta$.
Since the former is a constant at meta-train time, minimizing \eqref{eq: loss} is equivalent to maximizing the reward $\inner{\strategy^{t}}{\reward^t}$ the algorithm $\bandit_\theta$ gets at every $t\le T$ in a task $\task \sim \tasks$, given the previously observed rewards $\{\reward^\tau\}_{\tau=1}^{t-1}$.

Next, we will show two variants of the algorithm $\bandit_\theta$.

\begin{figure}[t]
\centering
    \includegraphics[width=0.48\textwidth]{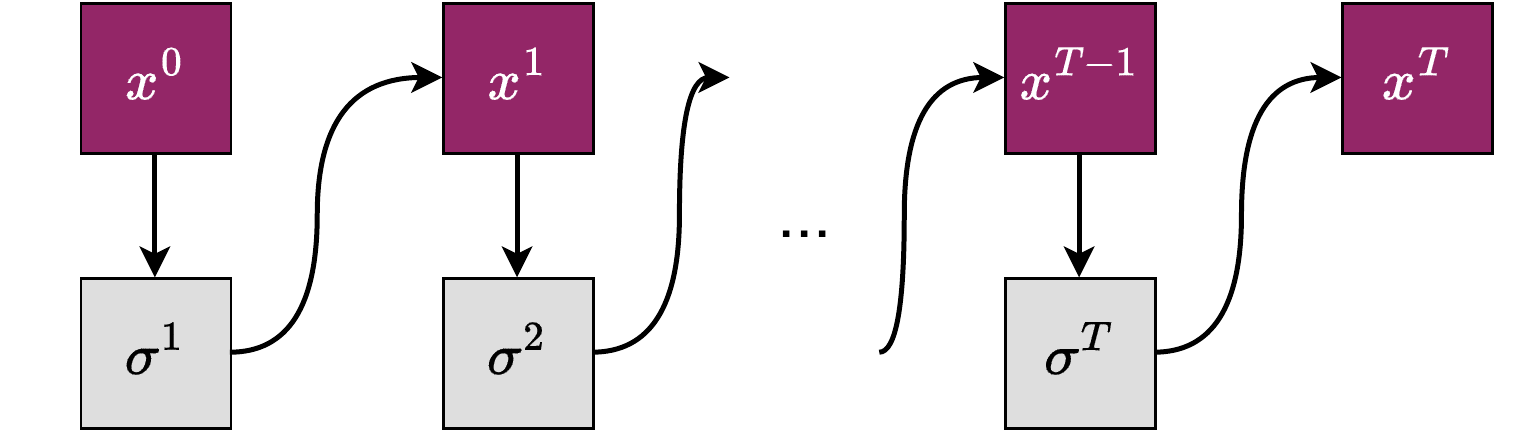}
\caption{
The sequence of strategies $\{\strategy^t\}_{t=1}^T$ submitted by an online algorithm and the rewards $\{\reward^t\}_{t=1}^T$ received from the environment. 
The reward $\reward^0 = \mathbf{0}$ initializes the algorithms to produce the first strategy $\strategy^1$.
}
\label{fig: strategy reward sequence}
\end{figure}

\begin{figure*}
\centering

\begin{subfigure}{0.45\textwidth}
    \includegraphics[width=1\textwidth]{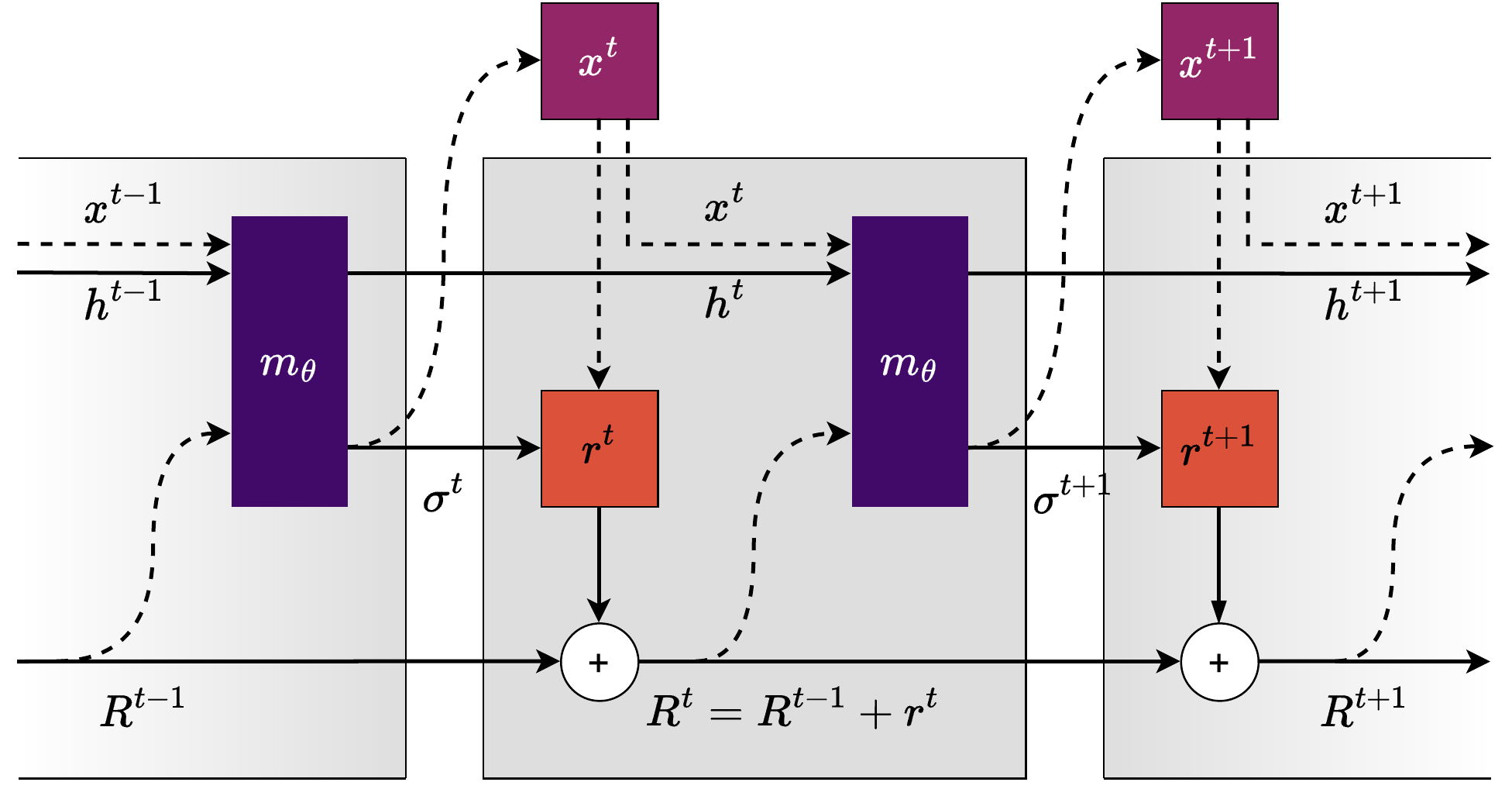}
    \caption{Neural online algorithm (NOA).}
    \label{Comp graph}
\end{subfigure}
\begin{subfigure}{0.45\textwidth}
    \includegraphics[width=1\textwidth]{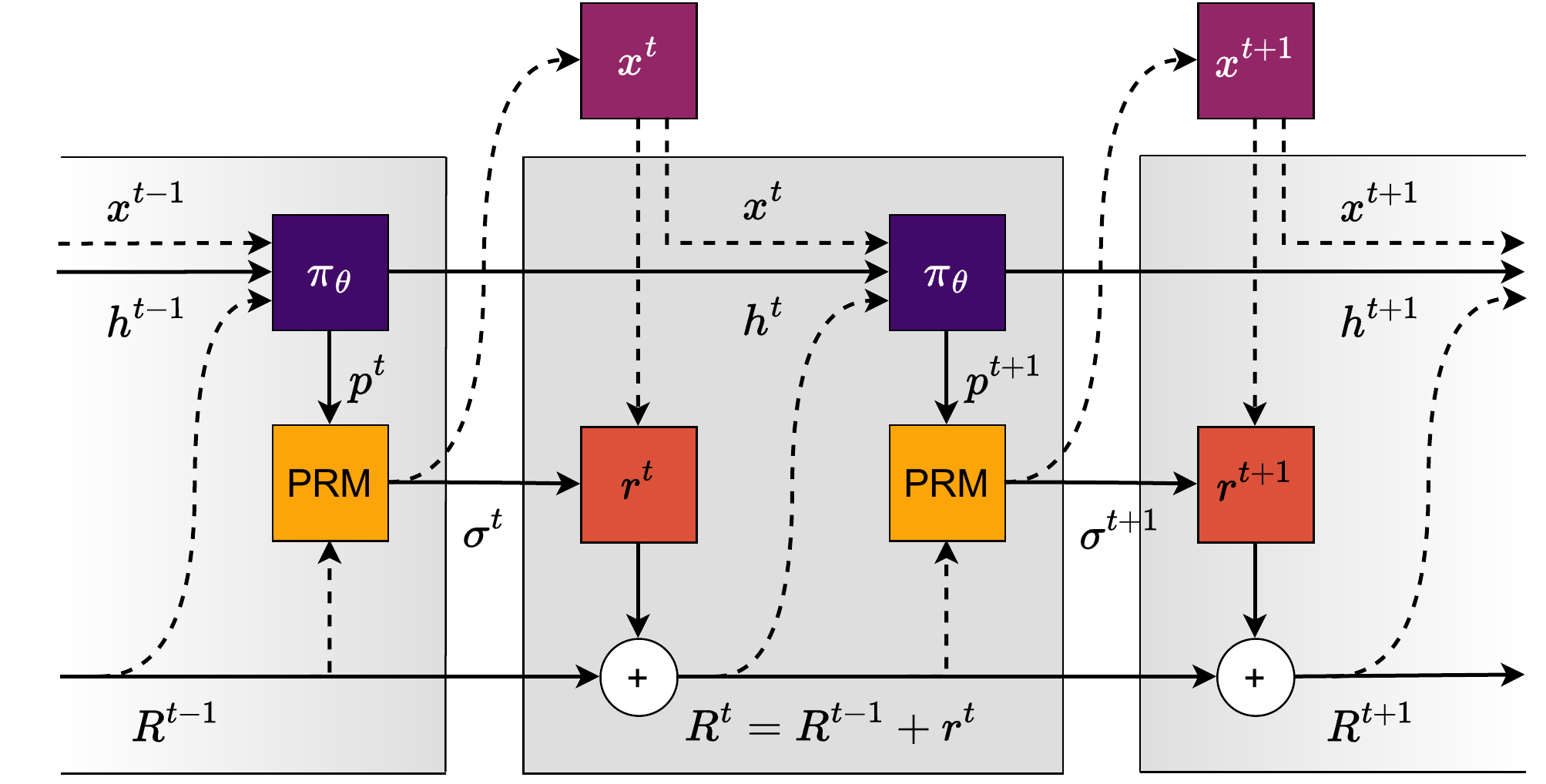}
    \caption{Neural predictive regret matching (NPRM).}
    \label{comp_graph2}
\end{subfigure}        
\caption{
Computational graphs of the proposed algorithms. The gradient flows only along the solid edges. The~$\vv{h}$ denotes the hidden state of the neural network. 
See also Figure~\ref{fig: strategy reward sequence} for visual correspondence of the strategy and reward sequence.
}
\end{figure*}

\subsection{Neural Online Algorithm}
The simplest option is to parameterize the online algorithm $\bandit_\theta$ to directly output the strategy $\strategy^t_\theta$.
We refer to this setup as neural online algorithm (NOA).

At the step $t$, the algorithm $\bandit_\theta$ receives as input\footnote{We also input additional contextual information, see Section~\ref{sec:experiments}.\label{fn: context info}} the rewards $\reward^t$ and cumulative regret $\cumregret^t$ and keeps track of its hidden state $\vv{h}^t$.
We estimate the gradient $\partial \mc L / \partial \theta$ by sampling a batch of tasks and applying backpropagation through the computation graph as shown in Figure~\ref{Comp graph}. The gradient originates in the final external regret $R^{{\rm ext}, T}$ and propagates through collection of regrets $\regret^{1 \dots T}$, the strategies $\strategy^{1 \dots T}$ and hidden states $\vv{h}^{0 \dots T-1}$. We don't allow the gradient to propagate through the rewards$^{\ref{fn: detach reward}}$ $\reward^{0 \dots T-1}$ or the cumulative regrets $\cumregret^{1 \dots T}$ entering the network. 
Thus, the only way to influence the earlier optimization steps is through the hidden states $\vv{h}^{0 \dots T-1}$ of the neural network.\footnote{This is similar to the ``learning to learn'' setup \citep{andrychowicz2016learning}\label{fn: detach regret}}

In our experiments, we observe strong empirical performance of NOA. However, NOA is not guaranteed to minimize regret.
This is because, similar to policy gradient methods, it is simply maximizing the cumulative reward $\mathbb{E}_{\task \sim \tasks}\left[\sum_{t=1}^T \reward^t\right]$,
which is not a sufficient condition to be a regret minimizing algorithm~\cite{blackwell1956analog}.

\subsection{Neural Predictive Regret Matching}
In order to get convergence guarantees, we turn to the recently introduced predictive regret matching (PRM)~\citep{farina2021faster}, see also Algorithm~\ref{algo:prm}. 
The PRM is an extension of regret matching (RM) \citep{hart2000simple} which uses an additional predictor~$\predictor: (\bullet) \rightarrow \R^{|A|}$. 
The algorithm has two functions, $\textsc{NextStrategy}$ and $\textsc{ObserveReward}$, which alternate over the sequence~\eqref{eq:strategy-reward-sequence}.
The predictor makes a prediction $\prediction^{t+1}$ of the next anticipated regret\footnote{Originally, the predictive regret was formulated in terms of next reward. However, for our application predicting next regret proved more stable as the network outputs don't mix.} 
$\regret^{t+1}$. The PRM algorithm 
incorporates $\prediction^{t+1}$ to compute the next strategy\footnote{Note the prediction can change the actual observed $\reward^{t+1}$, unless we are at a fixed point.} $\strategy^{t+1}$.
The RM algorithm can be instantiated as PRM with $\predictor = \vv{0}$. Unless stated otherwise, we use PRM with a simple predictor $\predictor: (\strategy^{t}, \reward^t) \to \prediction^{t+1} = \regret(\strategy^{t},\reward^t)$, i.e. it predicts the next observed rewards will be the same as the current ones.\footnote{This predictor was used in the original work.} 

We introduce neural predictive regret matching (NPRM), a variant of PRM which uses a predictor $\predictor_\theta$ parameterized by a recurrent neural network $\theta$. 
The predictor $\predictor_\theta$ receives as input$^{\ref{fn: context info}}$ the rewards $\reward^t$, cumulative regret $\cumregret^t$ and hidden state $\vv{h}^t$.
We train $\predictor_\theta$ to minimize Eq.~\eqref{eq: loss}, just like NOA. 
The computational graph is shown in Figure~\ref{comp_graph2}.
The output of the network $\prediction^{t+1}$ is used in $\textsc{NextStrategy}$ to obtain the strategy $\strategy^{t+1}$. Similar to NOA, the gradient $\partial \mc L / \partial \theta$ originates in the final external regret $R^{{\rm ext}, T}$ and propagates through the collection of regrets $\regret^{1 \dots T}$, the strategies $\strategy^{1 \dots T}$, the predictions $\prediction^{1 \dots T}$, and hidden states $\vv{h}^{0 \dots T-1}$. 
Again, we do not propagate the gradient through the rewards$^{\ref{fn: detach reward}}$ $\reward^{0 \dots T-1}$ or through the cumulative regrets $\cumregret^{1 \dots T}$ entering the network.$^{\ref{fn: detach regret}}$
Any time-dependence comes only through the hidden states~$\vv{h}^{0 \dots T-1}$.
It is interesting to note NPRM can learn to recover both RM and PRM as it receives all the information needed, i.e. $\reward$ and $\cumregret$.

\begin{algorithm}[t]
    \caption{Predictive regret matching\\ \citep{farina2021faster}}
    \label{algo:prm}
 
    \DontPrintSemicolon
    $\cumregret^0 \gets \vv{0} \in \R^{|A|}
    ,\ \ \reward^0 \gets \vv{0} \in \R^{|A|}
    $\;
    \Hline{}
    \Fn{\normalfont\textsc{NextStrategy}()}{
        $\displaystyle\predregret^t \gets [\cumregret^{t-1} + \prediction^{t}]^+$\;
        \textbf{if} $\|\predregret^t\|_1 > 0$ \hspace{0.2cm}
            \textbf{return} $\strategy^t \gets \predregret^t \ /\ \|\predregret^t\|_1$\;
        \textbf{else} \hspace{1.43cm}
            \textbf{return} $\strategy^t \gets $ arbitrary point in $\Delta^{|A|}$\hspace*{-1cm}\;
    }
    \Fn{\normalfont\textsc{ObserveReward}($\reward^{t}$)}{
        $\displaystyle\cumregret^t \gets \cumregret^{t-1} + \regret(\strategy^{t}, \reward^t)$\;
        $\displaystyle\prediction^{t+1} \gets \predictor(\reward^t)$\;
        \vspace{.5mm}
    }
\end{algorithm}

Importantly, we show that the cumulative regret of NPRM grows sub-linearly, making it a regret minimizer.

\begin{theorem}[Correctness of Neural-Predicting]
    \label{thm: correctness of nprm}
    Let $\alpha~\ge~0$, and $\predictor_\theta$ be a regret predictor with outputs bounded in $[-\alpha, \alpha]^{|A|}$.
    Then PRM which uses $\predictor_\theta$ is a regret minimizer.
\end{theorem}
\begin{proof}
    Since the reward $\reward$ for any action is bounded by the maximum utility difference $\maxutildiff$, the regret $\regret$ for any action is bounded by $2\maxutildiff$. Thus, for an arbitrary prediction $\prediction$ it holds
    \begin{equation*}
        \norm{\regret(\strategy, \reward) - \prediction} \le (2\maxutildiff+\alpha) |A|.
    \end{equation*}
    Using the PRM regret bound~\cite[Thm 3]{farina2021faster}, we obtain
    \begin{align*}
        R^{{\rm ext}, T} &\le
        \sqrt{2}\left(\sum_{t=1}^T\norm{\regret(\strategy^{t}, \reward^t) - \prediction^{t}}\right)^\frac{1}{2} \\
        &\le        
        \sqrt{2}\left((2\maxutildiff+\alpha) |A| T\right)^\frac{1}{2} 
        \in O\left(\sqrt{T}\right).
    \end{align*}
\end{proof}

As NPRM is regret minimizing regardless of the prediction $\prediction$, our network outputs $\prediction$ rather than strategy $\strategy$ as for NOA.
This allows us to achieve the best of both words -- adaptive learning algorithm with a small cumulative regret in our domain, while keeping the $O(T^{-1/2})$ worst case average regret guarantees. 
Note that $O(T^{-1/2})$ is the best achievable bound in terms of $T$ against a black-box~\cite{nisan2007algorithmic}.

\section{Experiments}
\label{sec:experiments}

\begin{figure*}[!t]
    \includegraphics[width=0.24\textwidth]{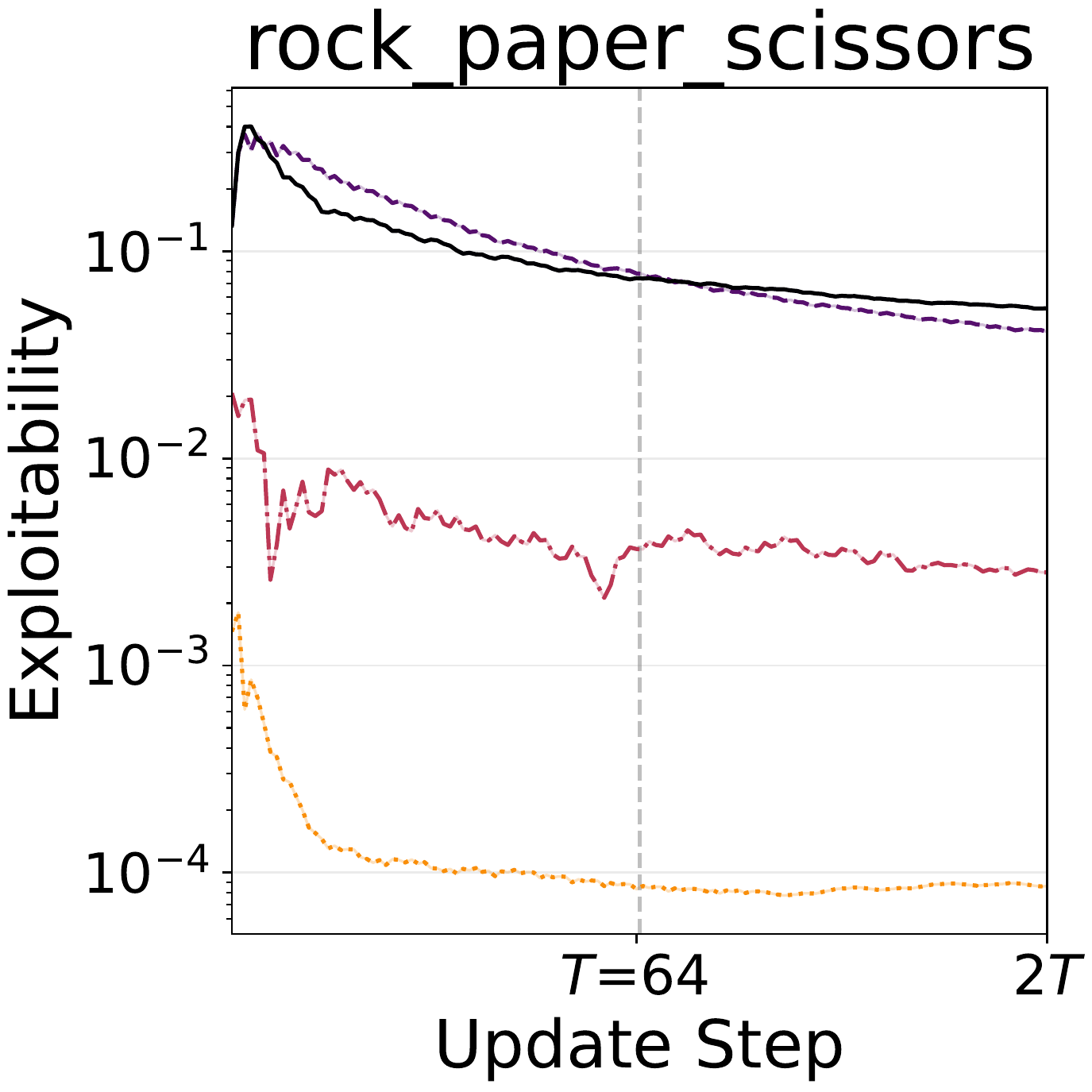}
    \includegraphics[width=0.24\textwidth]{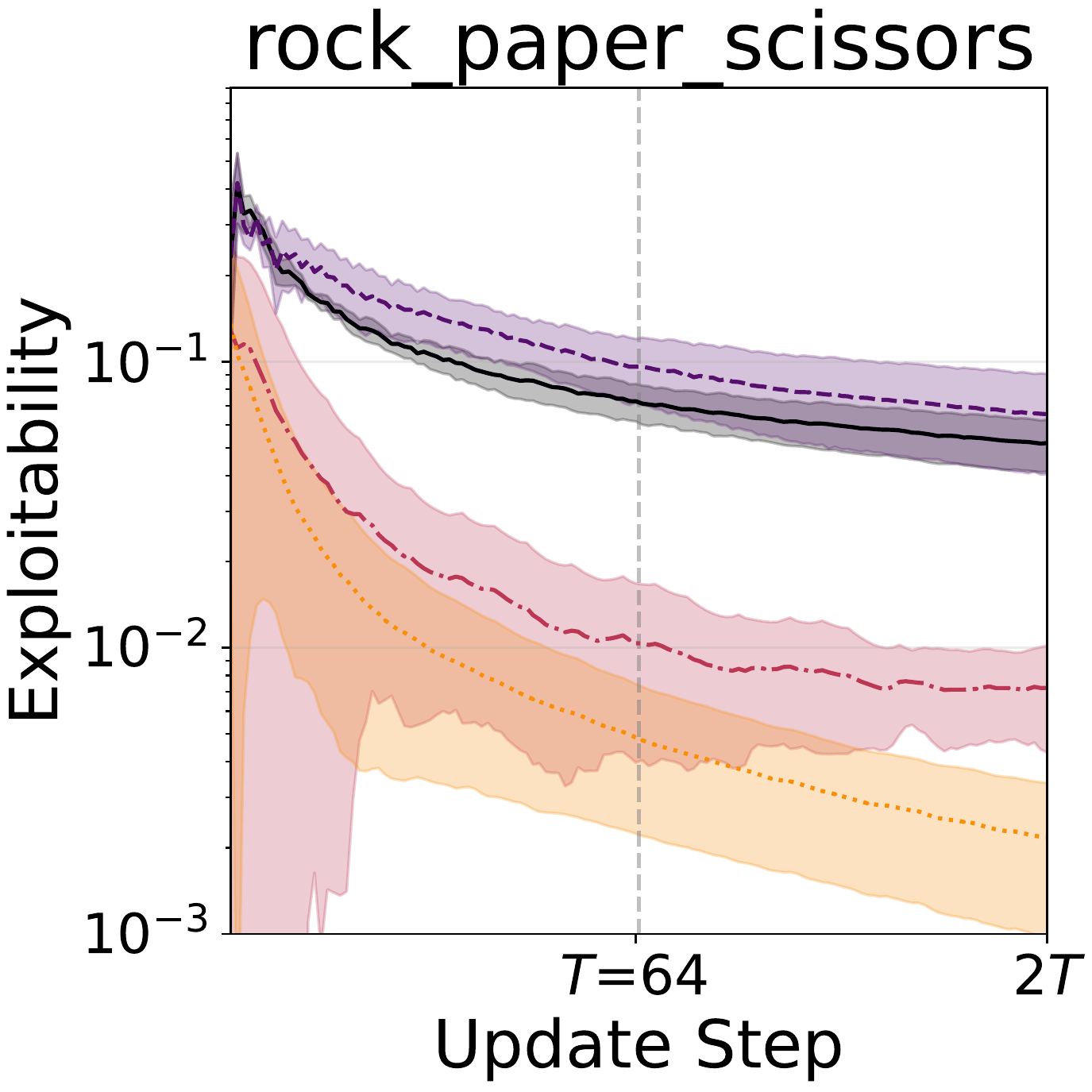}
    \includegraphics[width=0.24\textwidth]{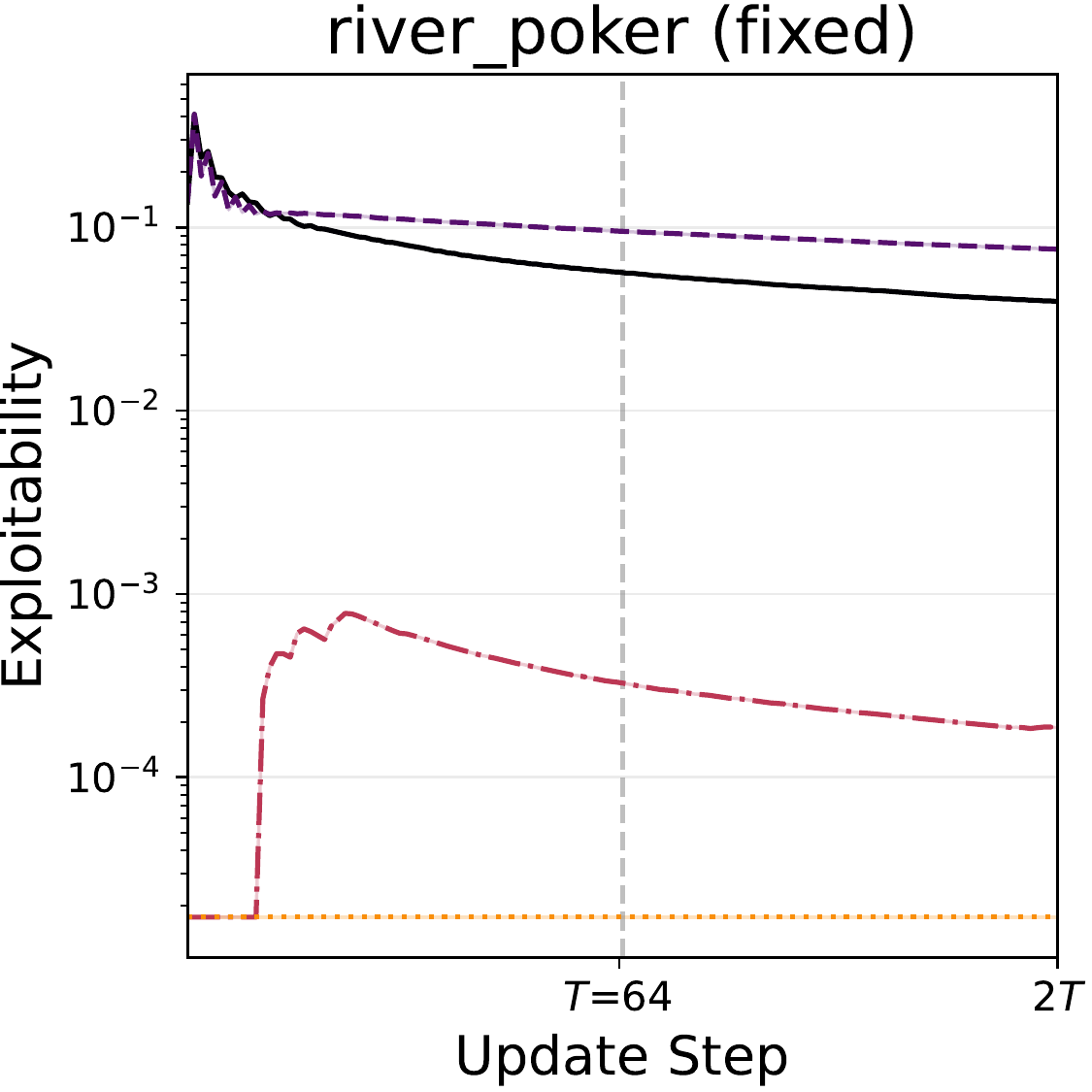}
    \includegraphics[width=0.24\textwidth]{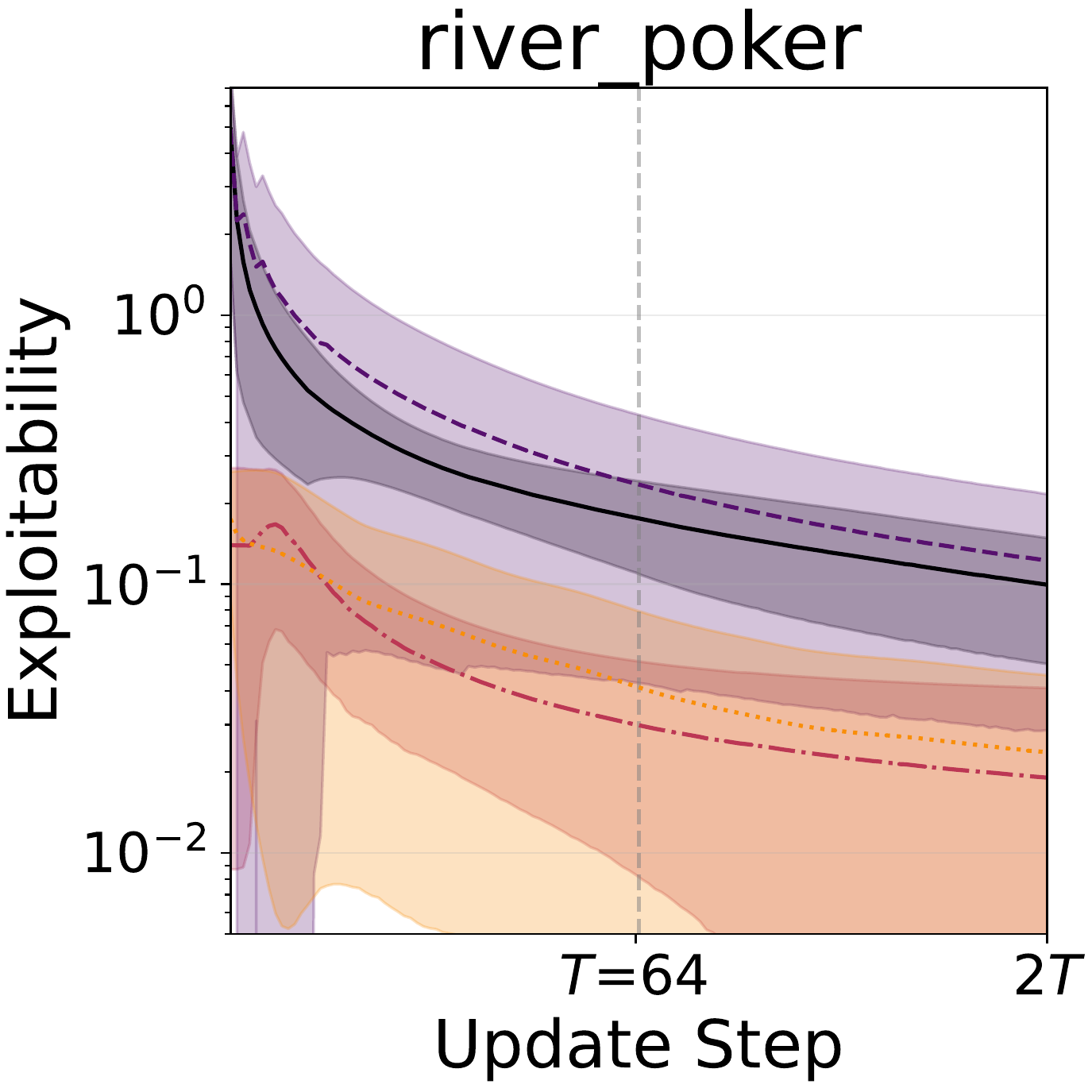}\\
    \centering
    \includegraphics[width=0.4\textwidth]{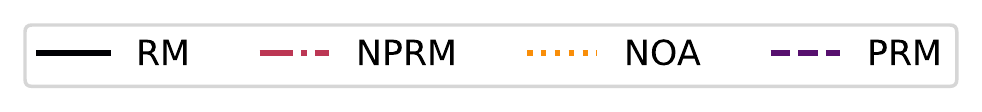}
\caption{Comparison of non-meta-learned algorithms (RM, PRM) with meta-learned algorithms (NOA, NPRM), on a small matrix game and a large sequential game and for a single fixed game versus a whole distribution over games. The figures show exploitability of the average strategy $\overline{\strategy}^t$. The y-axis uses a logarithmic scale. Vertical dashed lines separate two regimes: training (up to $T$ steps) and generalization (from $T$ to $2T$ steps). Colored areas show standard error for the sampled settings.
}
\label{fig: results}
\end{figure*}

We focus on application of regret minimization in games,
see Appendix~\ref{app: games} for their detailed description.
Specifically, we apply regret minimization to one-step lookahead search with (approximate) mini-max subgame value functions.
See also Section~\ref{sec: intro} for motivation of this approach.

For both NOA and NPRM, the neural network architecture we use is a two layer LSTM. For NOA, these two layers are followed by a fully-connected layer with the softmax activation. For NPRM, we additionally scale all outputs by $\alpha\ge2\maxutildiff$, ensuring any regret vector can be represented by the network. In addition to the last observed reward and the cumulative regret, the networks also receive contextual information corresponding to the player's observations.

We minimize objective~\eqref{eq: loss} for $T=64$ iterations over $512$ epochs using the Adam optimizer.\footnote{We use cosine learning rate decay from $10^{-3}$ to $3\cdot 10^{-4}$.} Other hyperparameters\footnote{Specifically, the size of the LSTM layer, the number of games in each batch gradient update, and the regret prediction bound $\alpha$.} were found via a grid search.
For evaluation, we compute exploitability of the strategies up to $2T=128$ iterations to see whether the algorithms can generalize outside of the horizon $T$ they were trained on and whether they keep reducing the exploitability.
We train and evaluate both NOA and NPRM and compare our methods against (P)RM. Our results are presented in Figure~\ref{fig: results}. 

The section is structured as follows.
First, we illustrate how our algorithms behave using a simple distribution of matrix games.
Next, we show how their performance extends to the sequential setting, where we evaluate on river poker.
To illustrate viability of our approach, we study the computational time reduction achieved by our algorithms.
Next, we demonstrate our algorithms are tailored to the training domain, and thus their performance can deteriorate out-of-distribution.
Finally, we discuss several possible modifications of our approach.

\subsection{Matrix Games}

\begin{figure*}[t!]
    \includegraphics[width=0.24\textwidth]{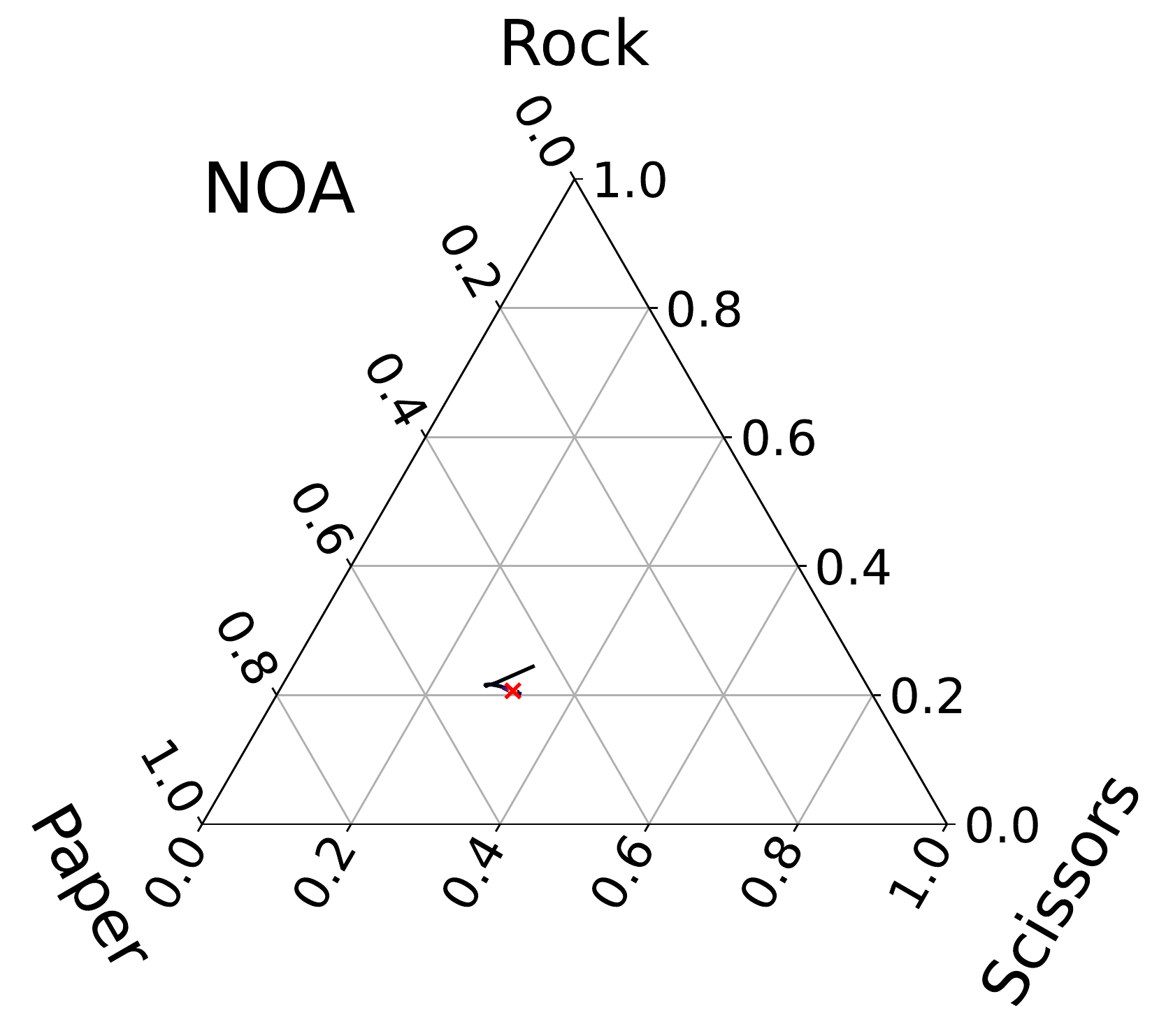}
    \includegraphics[width=0.24\textwidth]{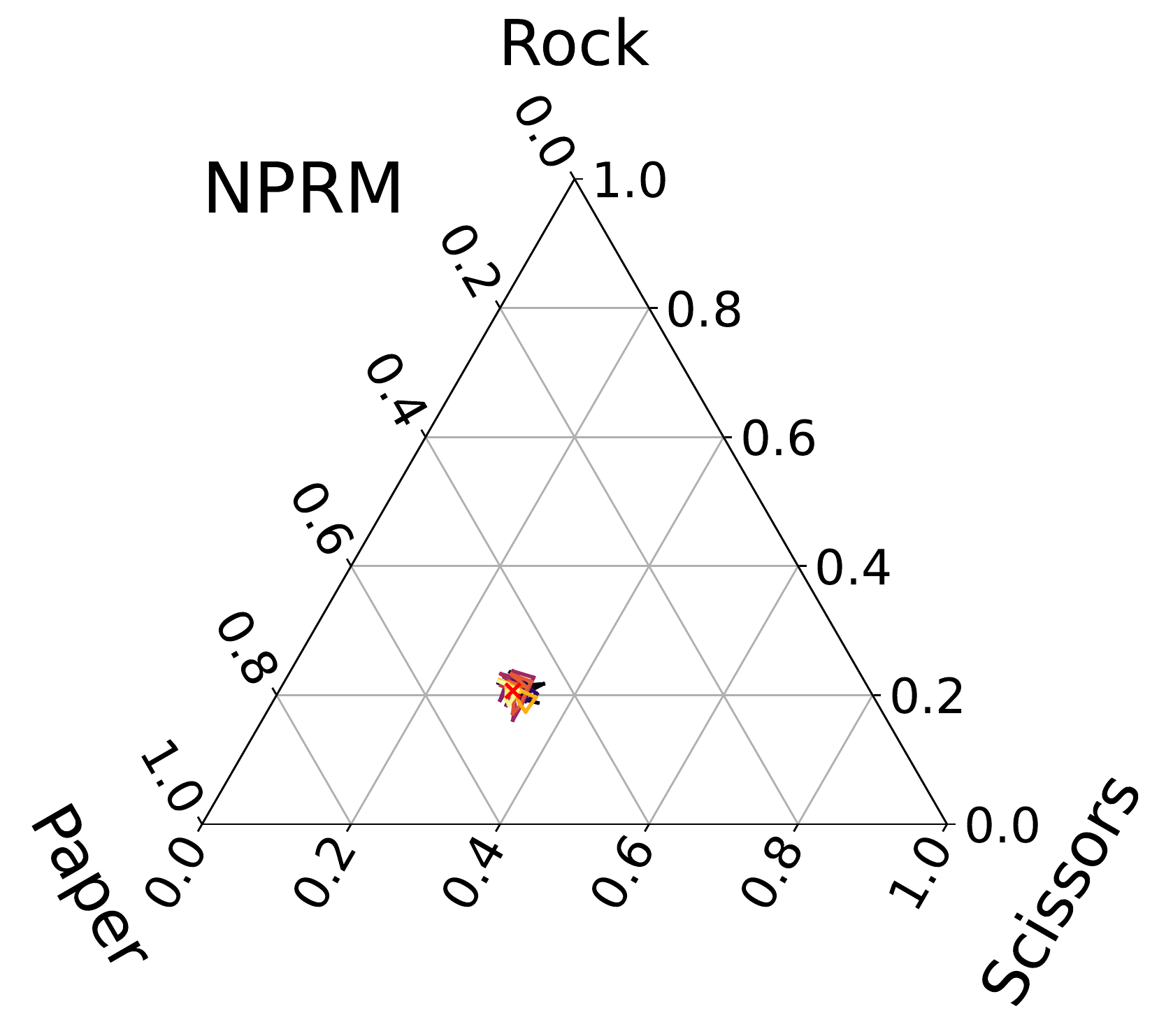}
    \includegraphics[width=0.24\textwidth]{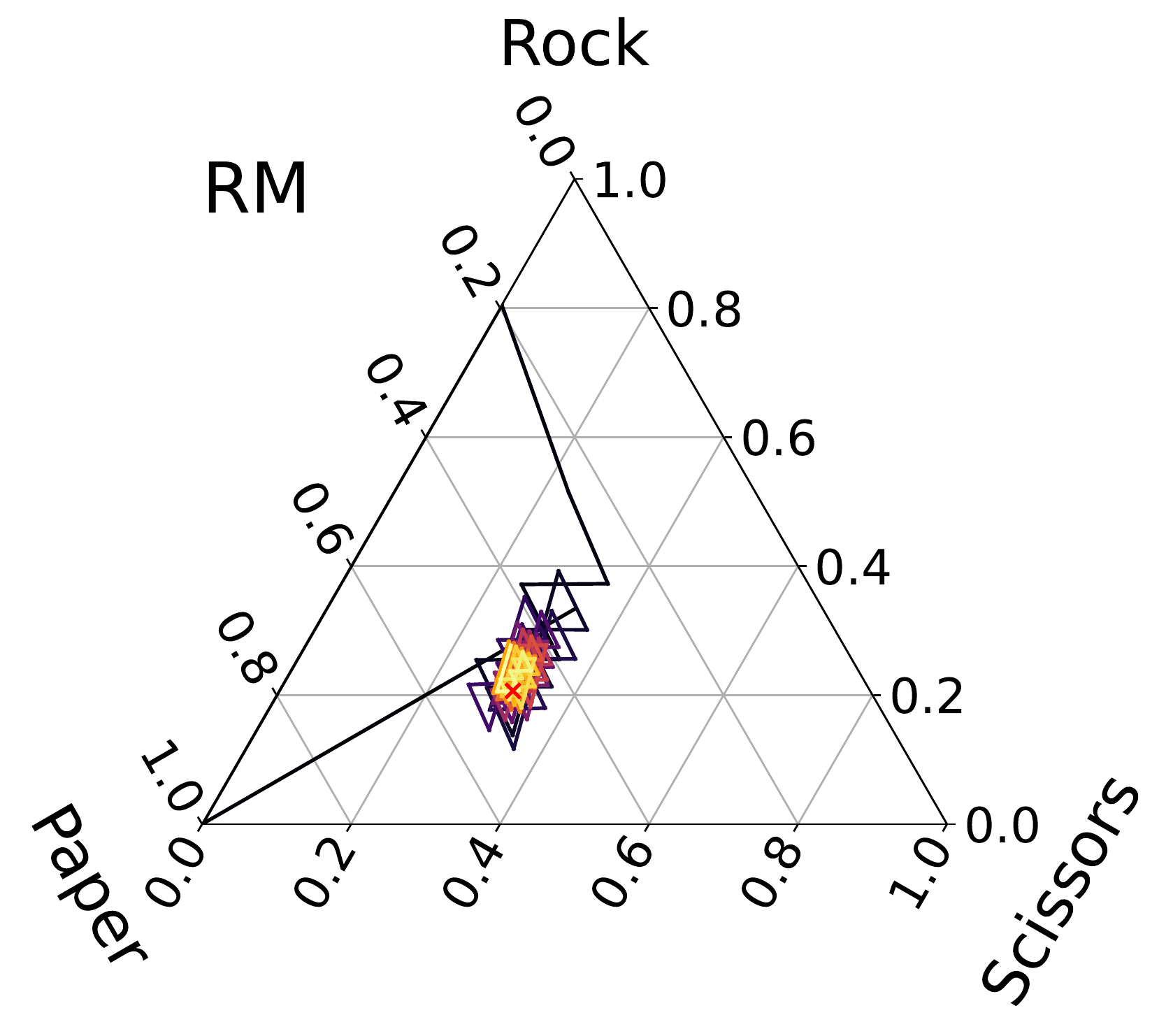}
    \includegraphics[width=0.24\textwidth]{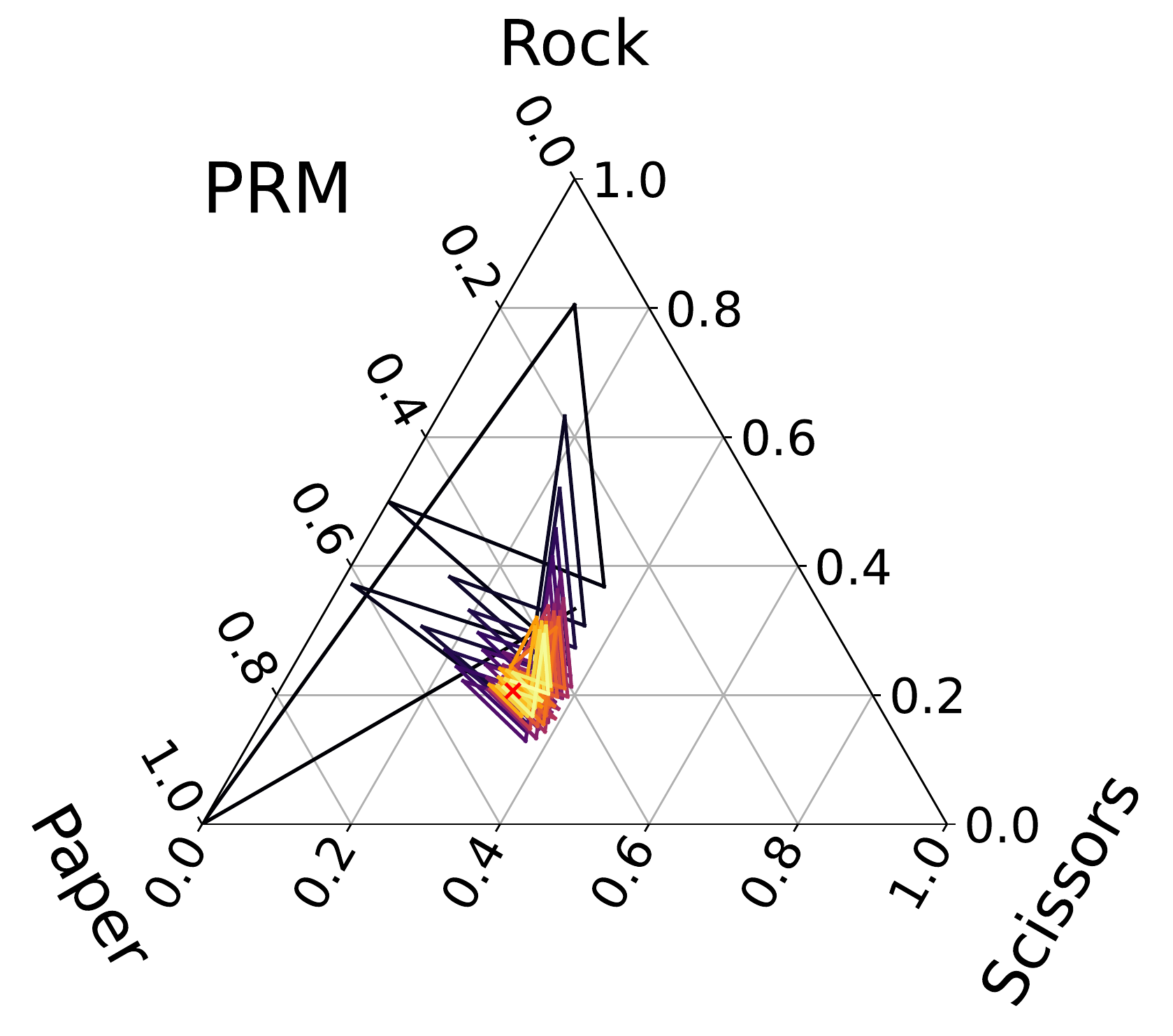}
    \\
    \includegraphics[width=0.24\textwidth]{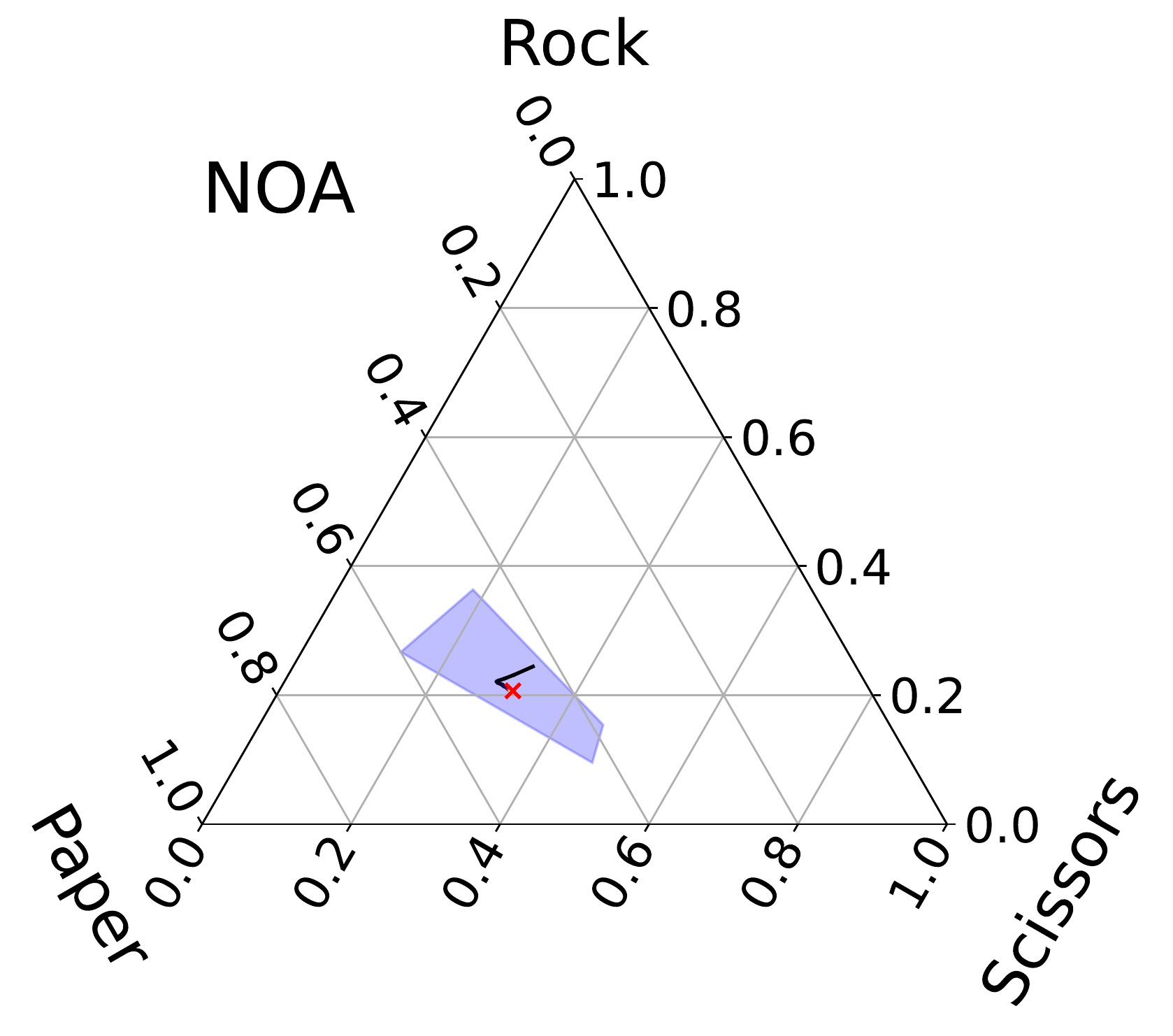}
    \includegraphics[width=0.24\textwidth]{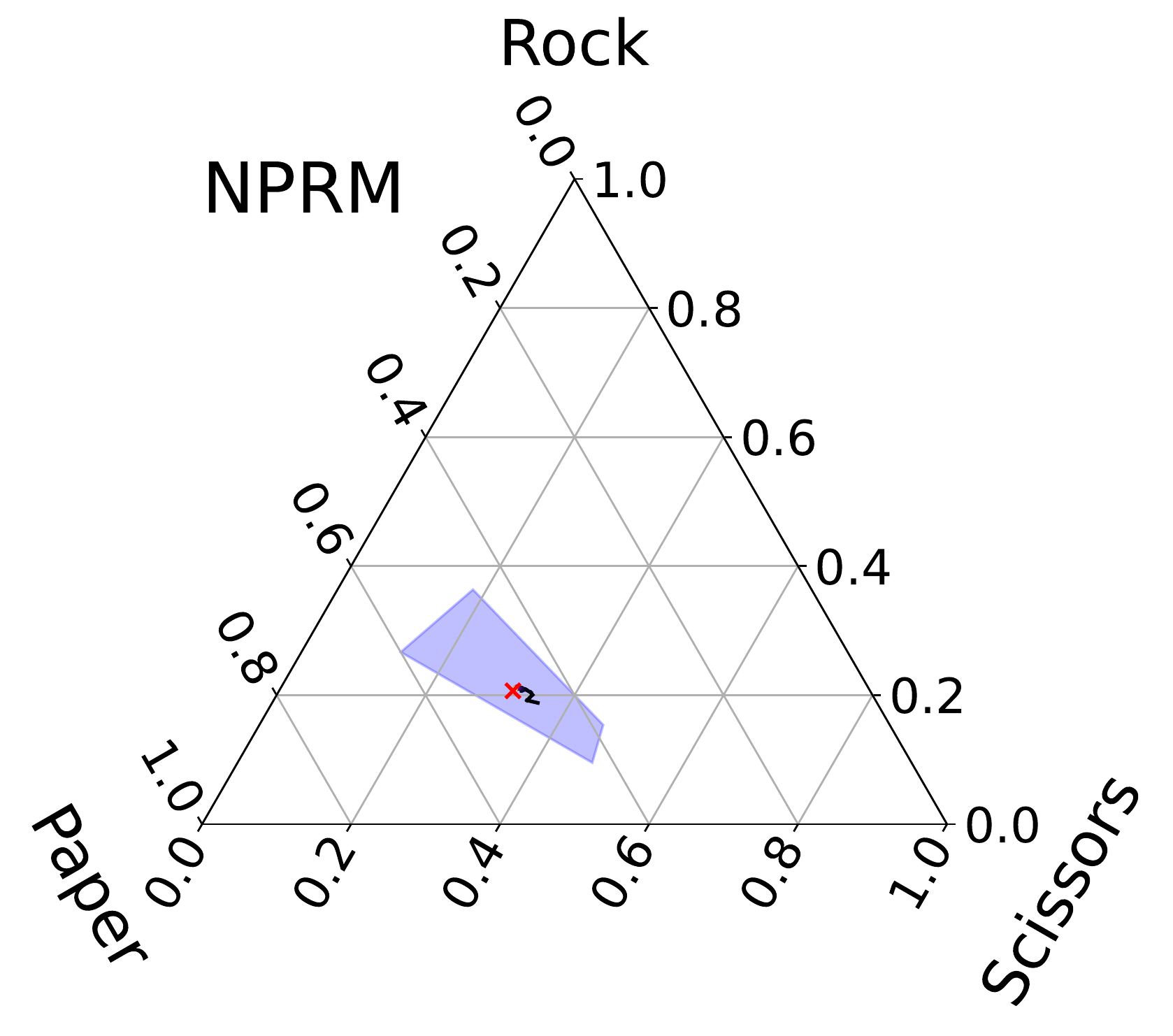}
    \includegraphics[width=0.24\textwidth]{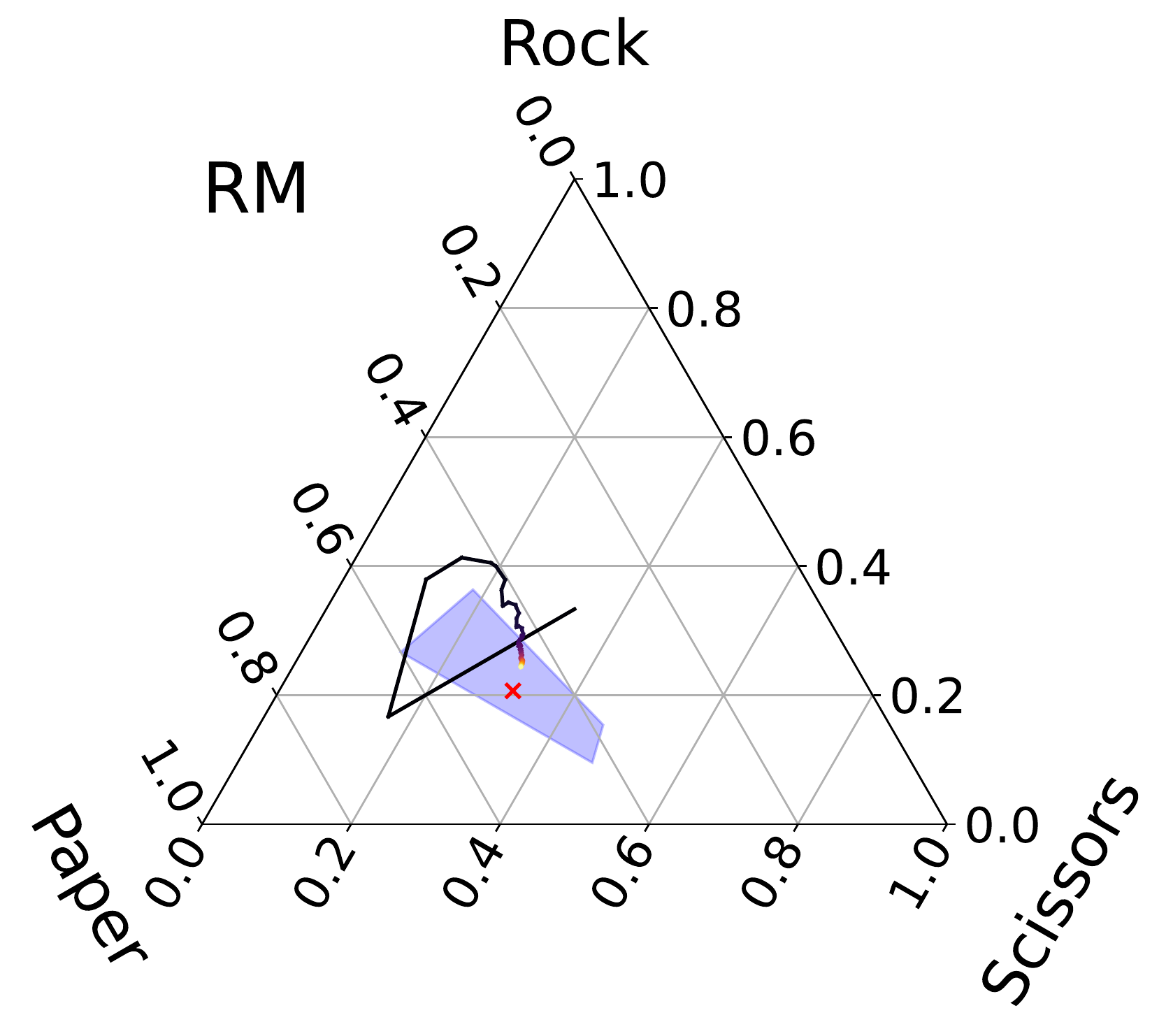}
    \includegraphics[width=0.24\textwidth]{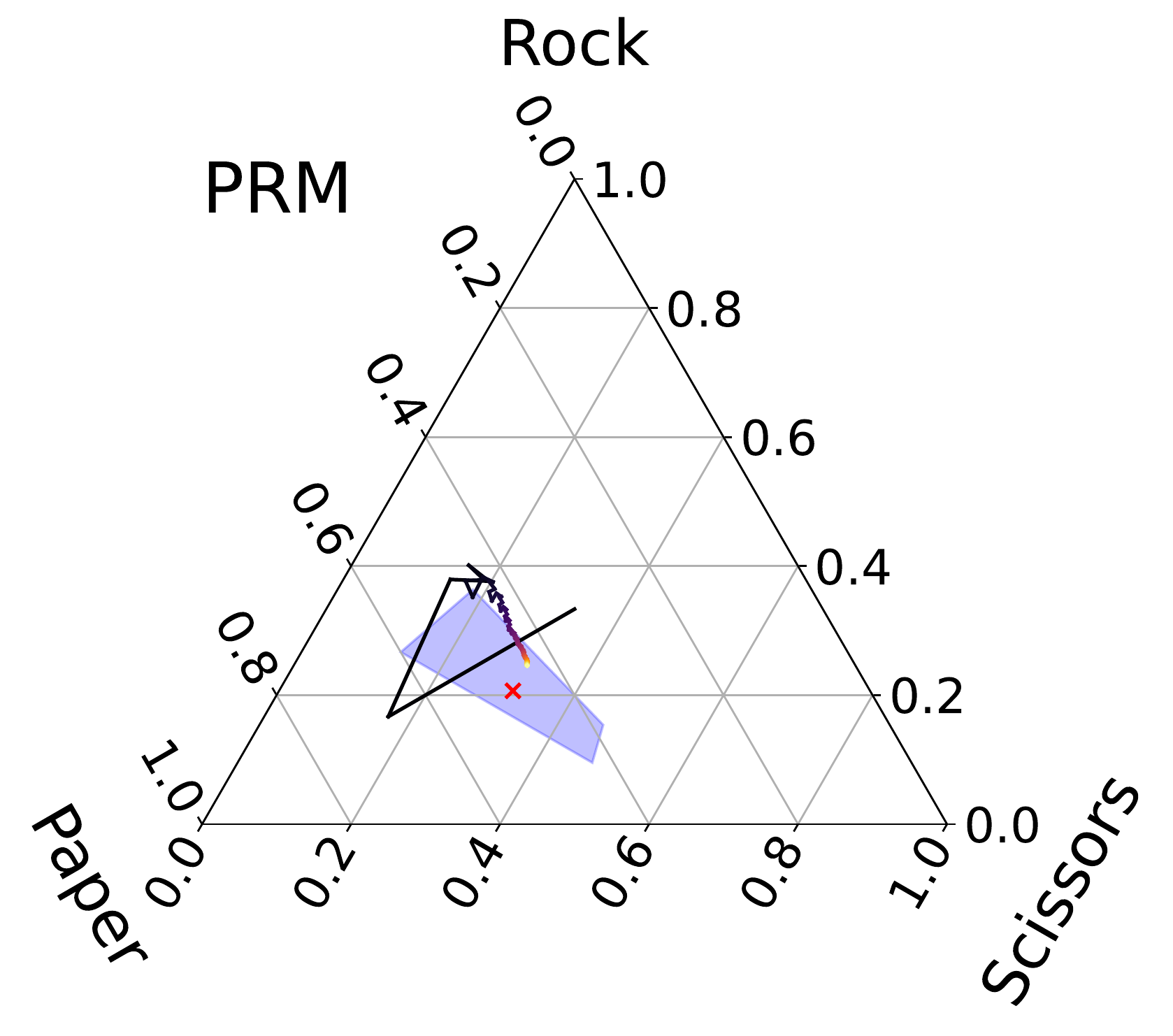}
\caption{For each algorithm, we show the trajectories of current strategies $\strategy^t$ (top row) and average strategies $\overline{\strategy}^t$ (bottom row) on {\tt rock\_paper\_scissors} (sampled) for $2T=128$ steps. 
The red cross shows the equilibrium of the sampled game. 
The trajectories start in dark colors and get brighter for later steps.
The blue polygon is the set of all equilibria in the distribution {\tt rock\_paper\_scissors}~(sampled), computed according to~\citep{BokHla2015a}. 
Notice how the strategies of our meta-learned algorithms begin in the polygon and refine their strategy to reach the current equilibrium. 
In contrast, (P)RM are initialized with the uniform strategy and visit a large portion of the policy space.
}
\label{fig: strategy evolution}
\end{figure*}

In the case of matrix games, a value function corresponds to playing against a best responding opponent.\footnote{A simultaneous-move matrix game can be formulated as a strategy-equivalent two step sequential game. The value function assumes optimal play by the opponent, i.e. a best response.}
We use a modification of the standard {\tt rock\_paper\_scissors} game and perturb two elements of the utility matrix to generate a distribution $\tasks$, see Appendix~\ref{app: matrix games}.

Our results are presented in Figure~\ref{fig: results}. 
First, we consider the distribution to have probability 1 for a single game, i.e. the game is fixed. In this setting, our algorithms can simply overfit and output a strategy close to a Nash equilibrium. 
Their convergence is very fast compared to (P)RM. 
Notice that NOA outperforms NPRM in this setting. 
We hypothesize there are two main reasons for this difference. First, NPRM is more restricted in its functional dependence. Second, the gradient of NPRM vanishes, resp. explodes when the cumulative regret is large, resp. small, making overfitting more challenging. 

Next, we sample games in the perturbed setting. 
Our methods keep outperforming (P)RM -- even after the horizon $T$ on which they were trained.
To further illustrate the differences between the meta-learned algorithms and (P)RM, we plot the current and average strategies selected by each algorithm in~Figure~\ref{fig: strategy evolution}. Both NOA and NPRM are initially close to the equilibrium and converge relatively smoothly. 
In contrast, (P)RM visit large portion of the policy space even in later steps, making the convergence slower.

Notice that RM exhibits similar performance to PRM both in matrix and sequential games.
This may seem surprising, as PRM was shown to be stronger than RM in self-play settings~\citep{farina2021faster}. 
However, the reason is that we minimize regret against an adversary rather than the self-play opponent.
PRM performs well when the last-observed reward is a good prediction of the next one. 
This is true in self-play, as the opponent does not radically change their strategy between iterations.
However, it is no longer the case when the values are coming from a value function where arbitrarily small modification of the input can lead to large changes of the output~\citep{schmid2021search}.

\begin{table}[t]
\centering
{\renewcommand{\arraystretch}{1.3}
\begin{tabular}{|c||c|c|c|c||}
\hline
Target   &  $ 4\cdot 10^{-1}$  &   $10^{-1}$   &   $6\cdot 10^{-2}$  &    $2\cdot 10^{-2}$     \\
\hline
\hline
RM   &   20  &   128   &  212  &   615      \\
\hline
PRM   &  36   &  158   &  261  &   793      \\
\hline
NOA   &  \textbf{1}   &   18    &  41   &   157      \\
\hline
NPRM   &   \textbf{1}  &   \textbf{16}   & \textbf{26}   &   \textbf{118}     \\
\hline
\end{tabular}
}
\caption{Number of steps each algorithm requires to reach target exploitability on {\tt river\_poker} (sampled) in expectation.}
\label{tab: target explo}
\end{table}

\begin{figure}[t!]
\centering
    \includegraphics[width=0.23\textwidth]{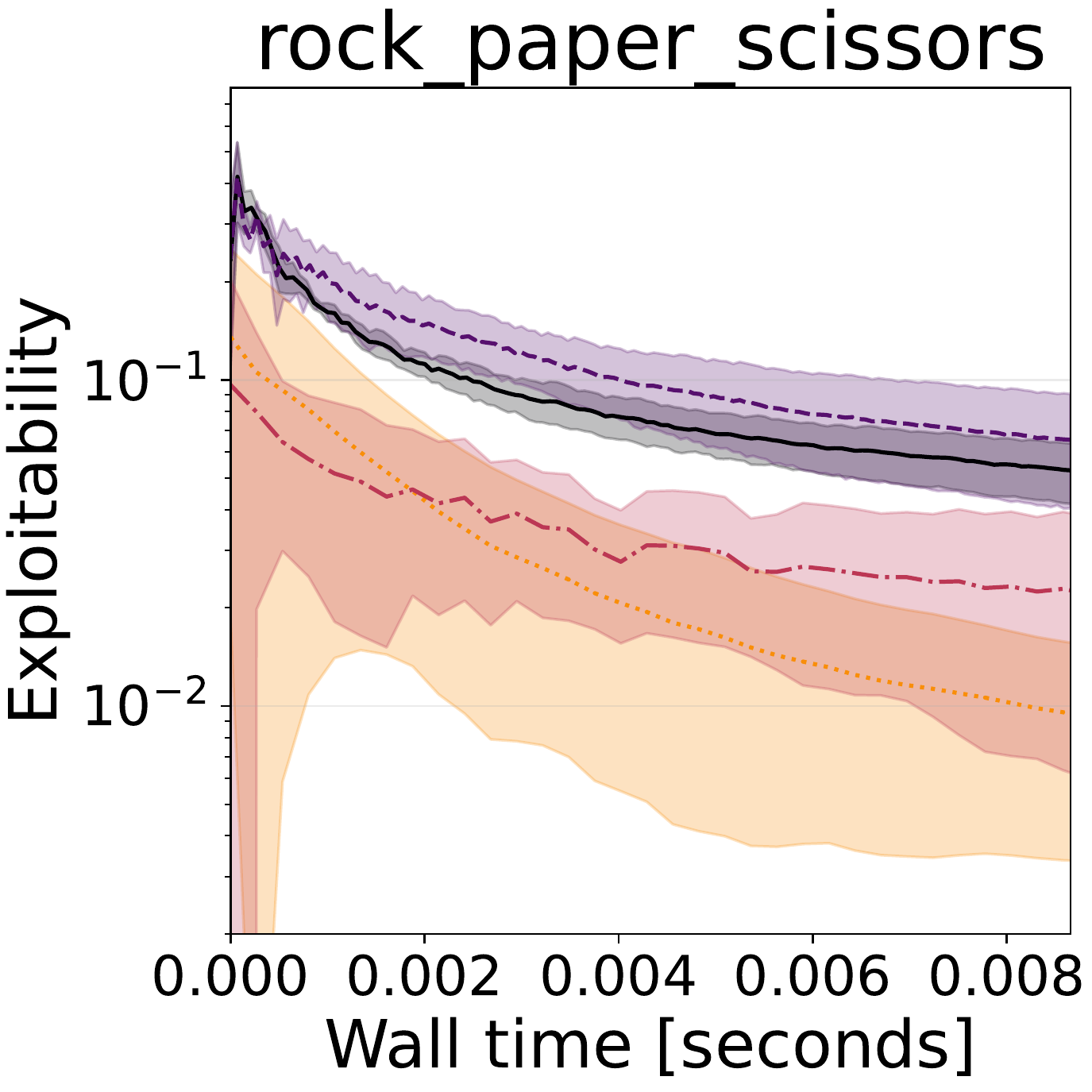}
    \includegraphics[width=0.23\textwidth]{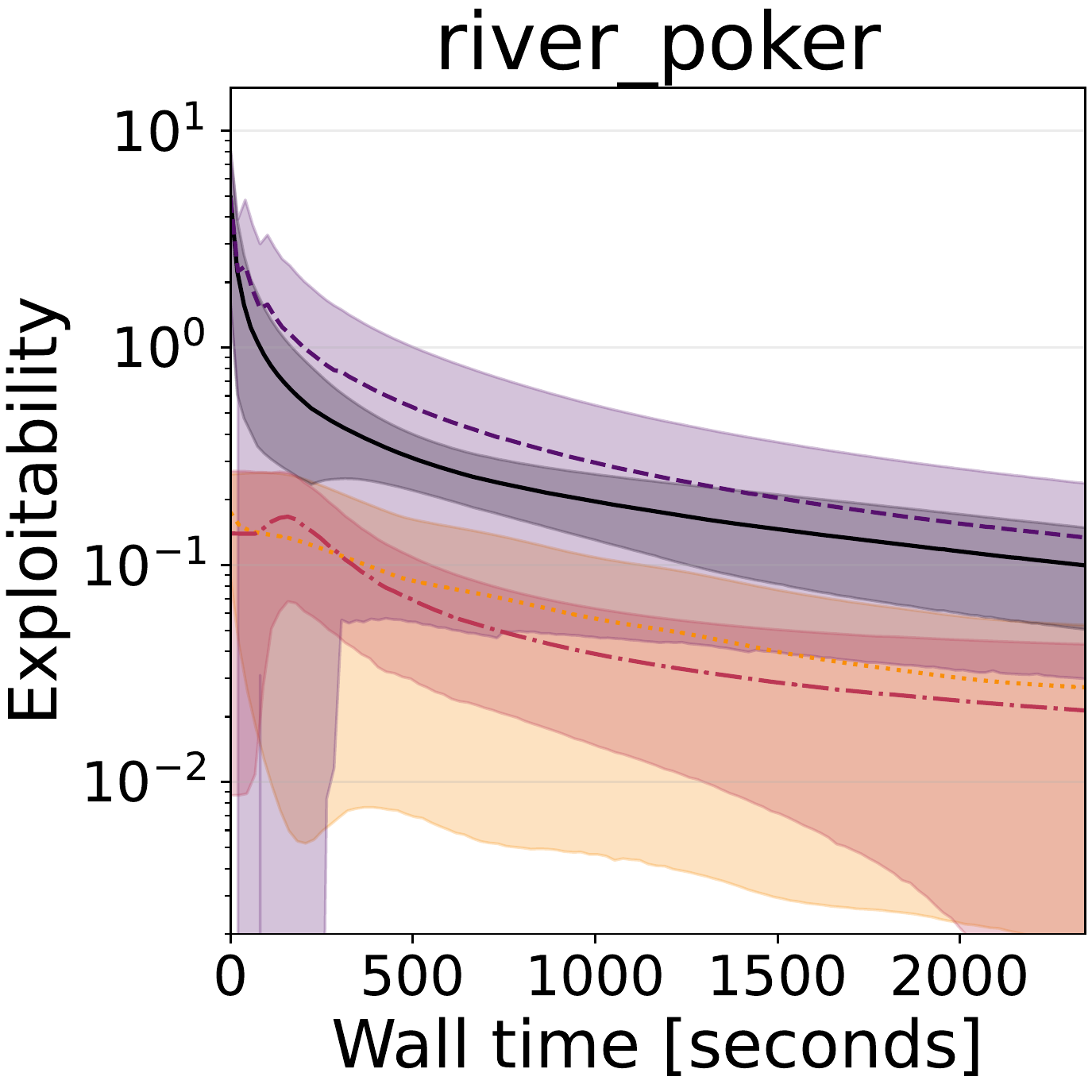}\\
    \centering
    \includegraphics[width=0.4\textwidth]{images/legend.pdf}
\caption{Comparison of regret minimization algorithms as a function of wall time, rather than number of steps shown in Figure~\ref{fig: results}.
}
\label{fig: roi walltime}
\end{figure}

\subsection{Sequential Games}
To evaluate our methods for sequential games, we use the public root state of {\tt river\_poker} -- a subgame of no-limit Texas Hold'em Poker with approximately $62$ thousand states. The distribution $\tasks$ is generated by sampling public cards
uniformly, while the player beliefs are sampled in the same way as in \citep{DeepStack}.
For the value function, we used 1,000 iterations of CFR$^+$.
See Appendix~\ref{app: sequential games} for more details. 

Our setup allows NOA and NPRM to learn to minimize regret in a contextualized manner, specific to each decision state. 
This is achieved by augmenting the input of the network by features corresponding to the player's observation at each state. 
In this case, the input is the the beliefs for both players and an encoding of private and public cards.

Our results are presented in Figure~\ref{fig: results}. 
We show that both NOA and NPRM are able to approximate an equilibrium of a fixed game very closely, often to higher precision than the solver. 
This manifests seemingly as a lower bound on exploitability for {\tt river\_poker} (fixed), see Appendix~\ref{app: approx value function} for details.
Importantly, even in the sampled setting, our algorithms greatly outperform (P)RM, reducing the exploitability roughly ten-times faster.
Just like in the matrix setting, PRM shows similar performance to RM, see previous section.

To further evaluate the improvements, we tracked how many steps it takes to reach a solution of specified target quality, see Table~\ref{tab: target explo}.
Both NOA and NPRM outperform (P)RM for all target exploitabilities, with better solutions requiring an order of magnitude less steps.

\subsection{Computational Time Reduction}

The reduction of the number of interactions with the environment may come at the expense of increasing the computational time, due to the overhead associated with calling the neural network. 
This time also depends on other factors, such as the selected domain, the available hardware, or the size of the network.
To assess the computational savings, we plot our results as a function of wall time in Figure~\ref{fig: roi walltime}.

On {\tt rock\_paper\_scissors}, the network overhead is noticeable, making each step of our methods about $4\times$ slower than (P)RM. 
Despite this, our methods keep outperforming (P)RM even after accounting for this extra cost. 
The offline meta-training was performed in about ten minutes. 

On {\tt river\_poker}, interacting with the environment is very expensive. 
Each interaction requires approximating optimal strategy\footnote{We wrote a costume solver for {\tt river\_poker} which outperforms other publicly available solvers. We made the solver available on https://github.com/DavidSych/RivPy.} in the subgame i.e. 1,000 iterations of CFR$^+$. 
Here, we observed the reduction in the number of steps translates well to the reduction of computational time.
For example, exploitability reached by NPRM after one minute would take RM, resp. PRM approximately 26, resp. 34 minutes to reach.
The meta-training on {\tt river\_poker} took two days.

We ran these experiments using a single CPU thread. 
As neural networks greatly benefit from using parallel processing, in some sense this can be seen as the worst-case hardware choice. 
Furthermore, for larger games than the ones considered here, each interaction is typically even more expensive.

\subsection{Out of Distribution Convergence}
As stated before, NOA is not guaranteed to minimize regret. However, NPRM is a regret minimizer even for games $\task'\sim\tasks' \neq \tasks$. Figure \ref{fig: out of distribution convergence} shows both methods trained to minimize regret on {\tt rock\_paper\_scissors} (sampled) and evaluated on {\tt uniform\_matrix\_game} (sampled). The results show that the performance of NOA deteriorates significantly. 
This is expected, as it aligns with the no-free-lunch theorems for optimization.
However, NPRM is able to keep minimizing regret even outside the domain it was trained on. In this case, it even outperforms (P)RM.

\begin{figure}[t]
    \includegraphics[width=0.23\textwidth]{images/matrix/rock_rock_394_152.pdf}
    \includegraphics[width=0.23\textwidth]{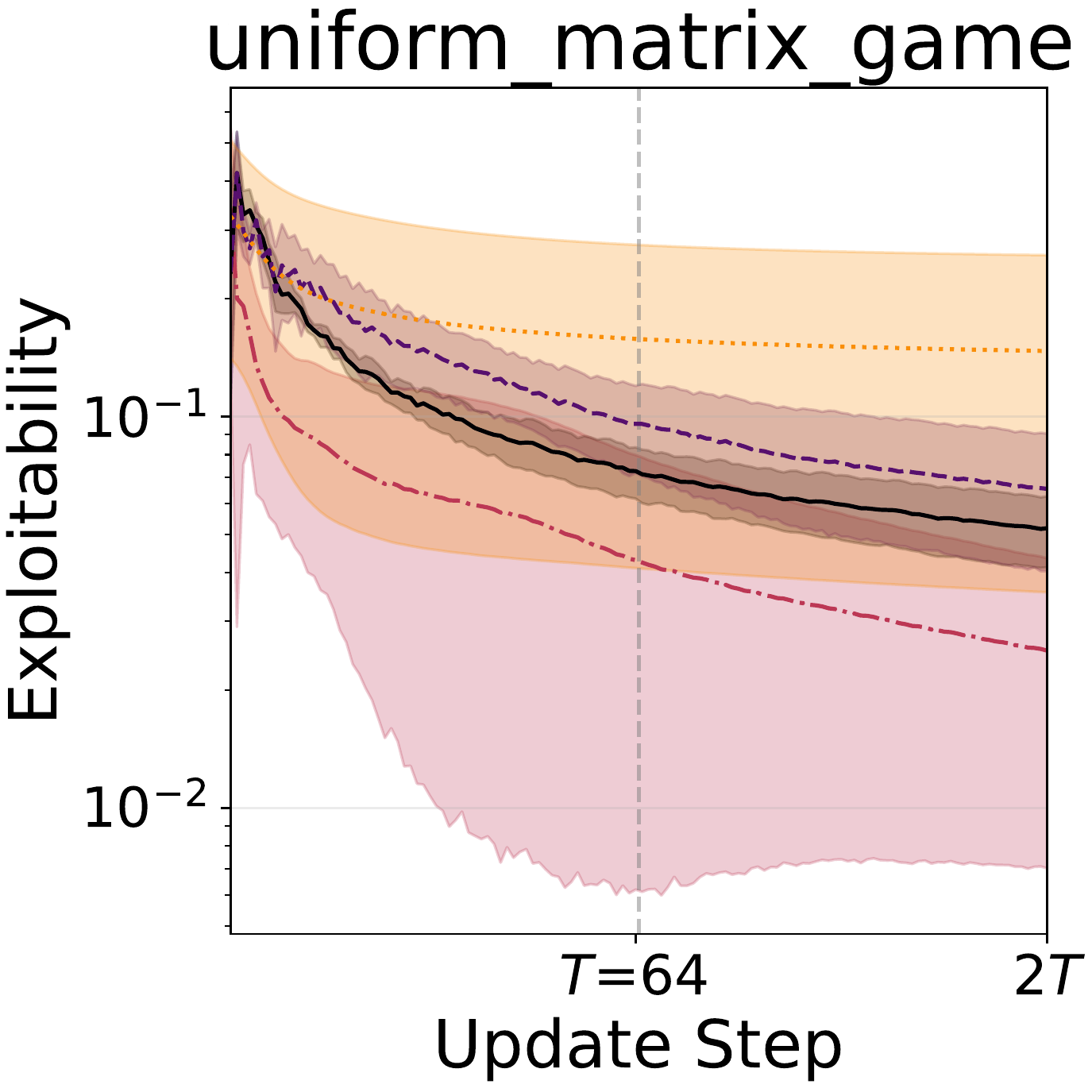}\\
    \centering
    \includegraphics[width=0.4\textwidth]{images/legend.pdf}
\caption{Comparison of the converge guarantees of NOA and NRPM. Both were trained on {\tt rock\_paper\_scissors}~(sampled).~Left~figure~shows NOA and NPRM can out outperform (P)RM on the distribution it was trained on. However, right figure shows that when evaluated on {\tt uniform\_matrix\_game} (sampled), the performance of NOA deteriorates significantly.}
\label{fig: out of distribution convergence}
\end{figure}

\subsection{Additional Experiments}
While the previous experiments made use of the common regret-matching-like setup, our meta-learning approach is more general. We investigated two modification based on previous methods. First, instead of aggregating the instantaneous regrets directly, we summed only the positive parts of said regrets, similar to (P)RM$^+$~\citep{tammelin2014solving}. Second, we used Hedge~\citep{freund1997decision} instead of regret matching to produce the strategy $\strategy$. We present results for both of these approaches in Appendix~\ref{app: extra experiments}, and Figure~\ref{fig: additional results}. 
Both meta-learned algorithms keep outperforming corresponding equivalents of (P)RM.

\section{Conclusion}

We introduced two new meta-learning algorithms for regret minimization in a new \enquote{learning not to regret} framework.
Our algorithms are meta-learned to minimize regret fast against a distribution of potentially adversary environments.
We evaluated our methods in games, where we minimize regret against an (approximate) value function and measure the exploitability of the resulting strategy.
Our experiments show that our meta-learned algorithms attain low exploitability approximately an order of magnitude faster than prior regret minimization algorithms.

In the future, we plan to extend our results to the self-play settings. 
We also plan to apply our methods with hindsight rationality \citep{morrill2021hindsight} for games which change over time.
This is also an opportunity to combine our offline meta-learning with the online meta-learning of~\citep{harris2022meta}.

\paragraph{Acknowledgements} 
The authors would like to thank Martin Loebl, Matej Morav\v{c}\'i{k}, Viliam Lis\'{y}, and Milan Hlad\'{i}k for their insightful comments.
This work was supported by the Czech Science Foundation grant no. GA22-26655S and CoSP Project grant no. 823748.
Computational resources were supplied by the project "e-Infrastruktura CZ" (e-INFRA LM2018140) provided within the program Projects of Large Research, Development and Innovations Infrastructures.

\clearpage

\bibliography{main}

\newpage
\clearpage
\appendix

\section{Games}
\label{app: games}
\subsection{Matrix Games}
\label{app: matrix games}
The {\tt rock\_paper\_scissors} is a matrix game given by
\begin{equation*}
    u_1 =
    -u_2 =
    \begin{pmatrix}
0 & -1 & 3+X \\
1 & Y & -1 \\
-1 & 1 & 0 
\end{pmatrix} ,
\end{equation*}
where the parameters $X,Y$ are set to zero for the {\tt rock\_paper\_scissors}~(fixed), and $X,Y\sim\mc U(-1, 1)$ for {\tt rock\_paper\_scissors}~(sampled). Note that the fixed variant is a biased version of the original game. We opted for this option to make the equilibrium strategy non-uniform, as in the original game (P)RM are initialized with the equilibrium policy.

The {\tt uniform\_matrix\_game} (sampled) is a 3$\times$3 matrix game with elements generated i.i.d. from $\mc U(-1, 1)$.

\subsection{Sequential Game}
\label{app: sequential games}

For {\tt river\_poker}, we use the endgame of no-limit Texas Hold'em Poker with all public cards revealed. The currency used is normalized such that the initial pot of each player is one. The total budget of each player is set to one-hundred times that amount, which implies there are 61,617 information states in total. 
To create a distribution, we sample the five public cards, and the beliefs for both players in the root of the subgame. The public cards are sampled uniformly, while the beliefs are sampled in the same way as \citep{DeepStack}. The algorithms presented in the main text are used only in the public root state, and the optimal strategy in the rest of the game is approximated via 1,000 iterations of self-play CFR$^+$. The corresponding approximate counterfactual values were used as rewards $\reward^t$. The exploitability of the optimal extension was obtained by approximating the game value (again via CFR$^+$), and subtracting the value in the root given the average strategy of each algorithm. We opted to use CFR$^+$ due to its strong empirical performance on poker games \citep{bowling2015heads}.
Compared to value functions represented for example by a neural network, this approach offers strong guarantees and replicability.

\section{Approximate Value Function Error}
\label{app: approx value function}
To approximate the exploitability, one needs to approximate both the value function, and the game value of the {\tt river\_poker}.
As stated above, we used CFR$^+$ in self-play to approximate both.
During evaluation of {\tt river\_poker} (fixed), we used 10$\times$ more CFR$^+$ iterations than in training, to improve the approximation of both game value and the value function.
In this case, when used to approximate the game value, it yields a solution with a two-player Nash gap of approximately $1.6 \cdot 10^{-5}$.
The error is of similar magnitude as the exploitability observed in Figure~\ref{fig: results} for {\tt river\_poker} (fixed) and it explains the apparent lower bound.

During evaluation, we observed rapid changes of performance of NPRM when we changed the number of CFR$^+$ iterations.
Since it was trained using 1,000 iterations, the evaluation is effectively out-of-distribution. 
We hypothesise that the reason why NPRM struggles more than NOA is that PRM is very sensitive to small changes in the prediction when the cumulative regret is small.
This makes not only the training challenging, but may also explain the sudden degradation in performance observed on {\tt river\_poker} (fixed).

\section{Alternative Setups}\label{app: extra experiments}

\begin{figure}[t!]
    \includegraphics[width=0.23\textwidth]{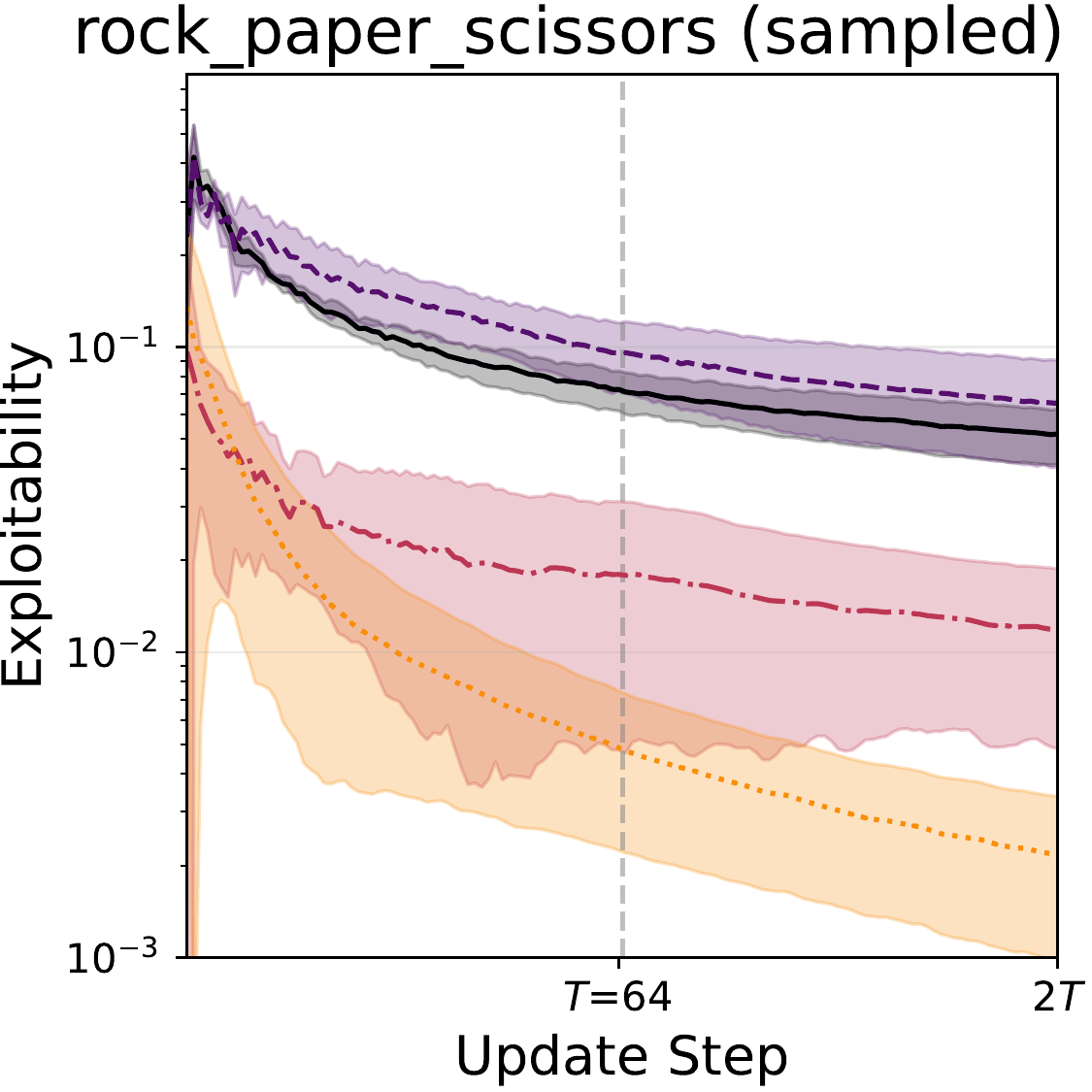}
    \includegraphics[width=0.23\textwidth]{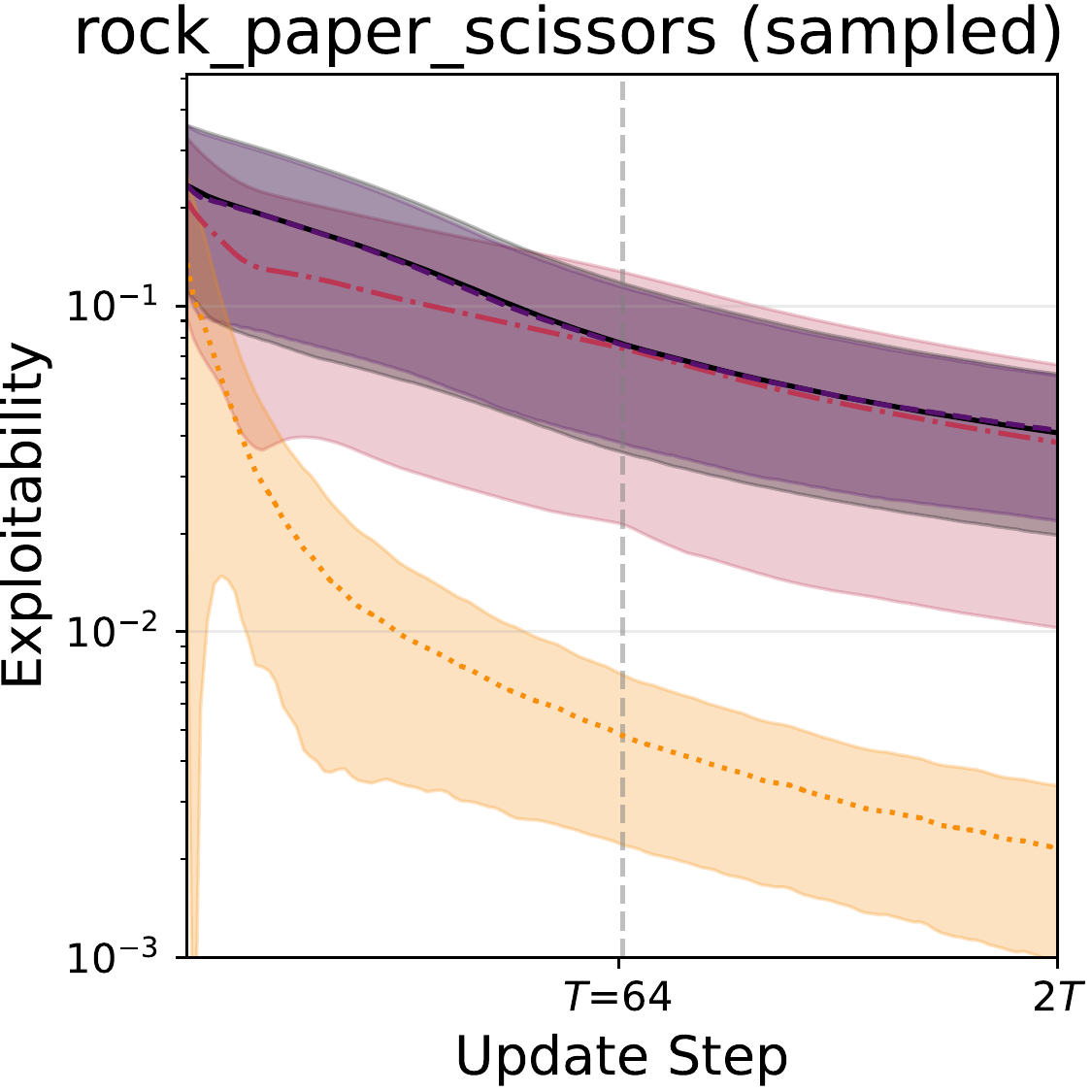}\\
    \\
    \includegraphics[width=0.23\textwidth]{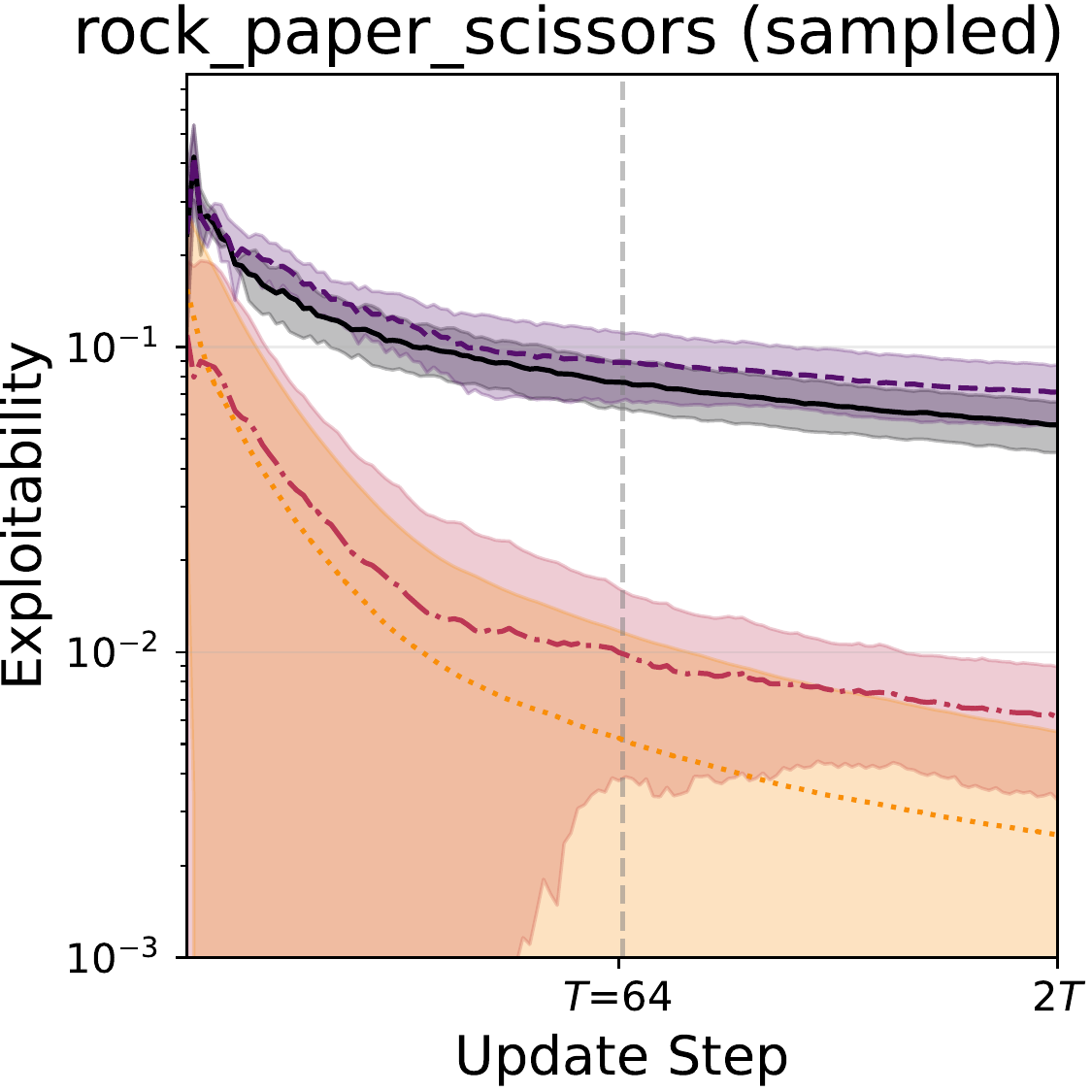}
    \includegraphics[width=0.23\textwidth]{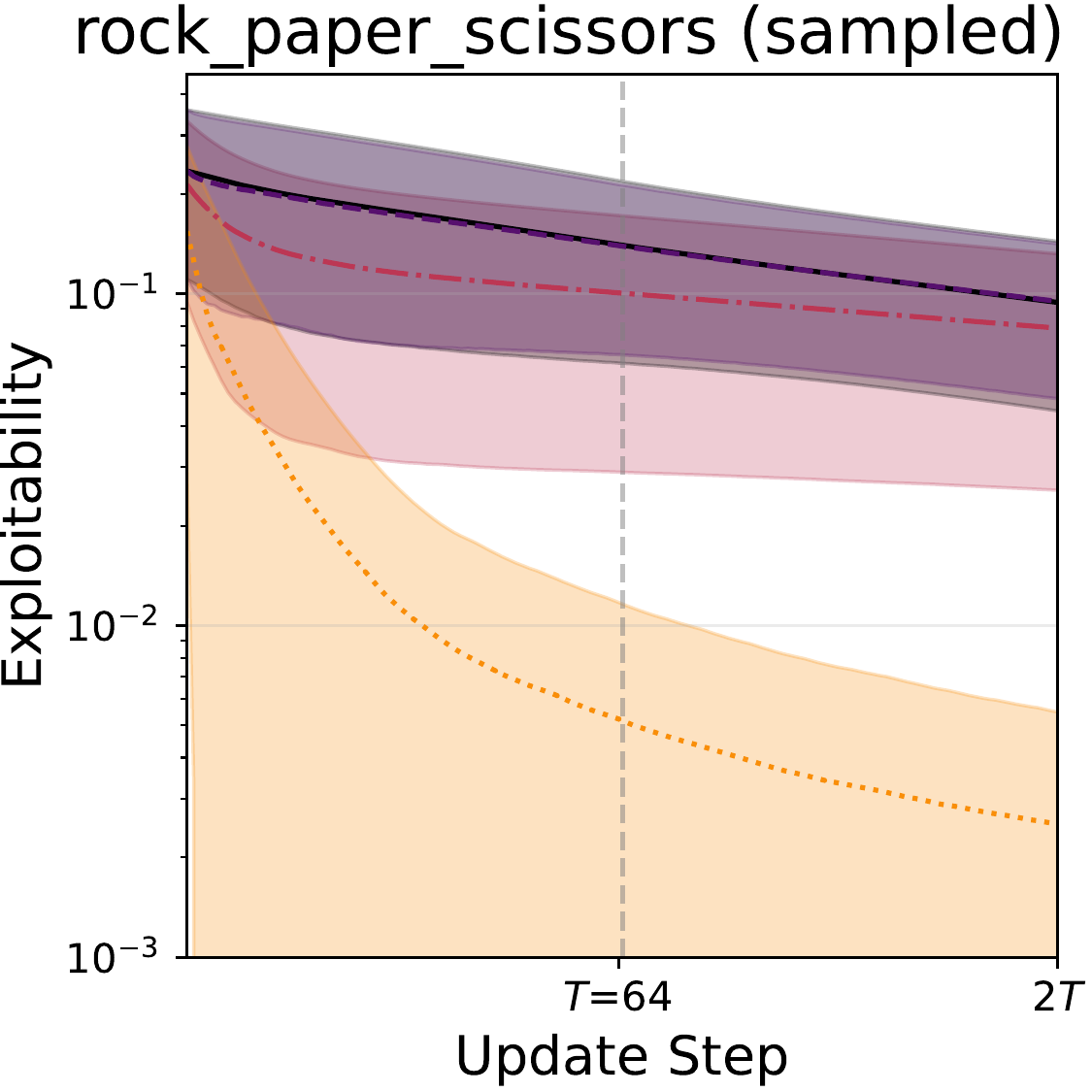}\\
    \centering
    \includegraphics[width=0.3\textwidth]{images/legend.pdf}
\caption{Comparison of regret minimization algorithms against best responding opponent. 
The figures show exploitability of the average strategy $\overline{\strategy}$. 
Vertical dashed line separates the training (up to $T$ steps) and generalization (from $T$ to $2T$ steps) regimes. 
The bottom row aggregates only positive regret, while the right column uses Hedge.}
\label{fig: additional results}
\end{figure}

In this section, we show our methods can be combined with two popular methods for regret minimization. First, similar to CFR$^+$~\citep{tammelin2014solving}, we aggregate only positive parts of the regret $\vv{R}^t = [\vv{R}^{t-1}+\vv{r}^t]^+$.
Second, we use Hedge~\citep{freund1997decision} to produce the strategy
\begin{equation*}
    \vv{\strategy}^t = 
    \frac{e^{\beta(R^t + p^t)}}{\sum_{a\in A}e^{\beta(R^t + p^t)}},
    \hspace{3ex}
    \beta = \sqrt{\frac{2\log(|A|)}{T}},
\end{equation*}
for NPRM, RM and PRM. We train and evaluate all algorithms on {\tt rock\_paper\_scissors}~(sampled). Our results are presented in Figure~\ref{fig: additional results}. 

Aggregating only positive regret seems to improve the performance of NPRM, and hinder NOA. Since RM$^+$ was observed to outperform RM on similar games~\citep{tammelin2014solving}, it may be the case this helps NPRM as well. In contrast, NOA receives less information through $\vv{R}$. 
Hedge exhibits slower convergence in general, and severely decreases the performance of NPRM. Interestingly, higher values of $\alpha$ perform better with Hedge. This corresponds to the fact that in order to get a strategy far from uniform, the cumulative regret needs to be large compared to when regret matching is used.

Note that it is known the temperature $\beta$ can be tuned, leading to improved performance~\citep{burch2018time}.
This can also be done within our meta-learning approach, by allowing the network to output $\beta$ directly.

\end{document}